\documentclass[11pt,letterpaper]{article}
\PassOptionsToPackage{margin=1in}{geometry}
\usepackage[T1]{fontenc}
\usepackage[utf8]{inputenc}
\usepackage[english]{babel}
\usepackage{lmodern,microtype,amsmath,amssymb,amsthm,mathtools}
\usepackage{geometry}
\usepackage{booktabs,array,enumitem,tikz}
\usepackage{algorithm,algpseudocode,needspace}
\floatstyle{ruled}
\restylefloat{algorithm}
\usepackage[hidelinks]{hyperref}
\DeclareRobustCommand{\algname}[1]{\textnormal{\textsc{#1}}}
\hypersetup{pdftitle={Almost Lossless Transfer for Cuttable Distributed Problems},
  pdfauthor={Wojciech Wawrzyniak}}
\setlist{itemsep=.2em}
\newtheorem{theorem}{Theorem}[section]
\newtheorem{lemma}[theorem]{Lemma}
\newtheorem{fact}[theorem]{Fact}
\newtheorem{observation}[theorem]{Observation}
\newtheorem{proofclaim}{Claim}[theorem]

\newtheorem{corollary}[theorem]{Corollary}
\newtheorem{maintheorem}[theorem]{Theorem}
\theoremstyle{definition}
\newtheorem{definition}[theorem]{Definition}

\newcommand{\eps}{\varepsilon}

\newcommand{\C}{\mathcal C}
\newcommand{\D}{\mathcal D}
\newcommand{\PP}{\mathcal P}
\newcommand{\N}[2]{N^{#1}[#2]}
\newcommand{\OPT}{\operatorname{OPT}}
\newcommand{\Bsingle}{B_1}
\newcommand{\Bshift}{B_2}
\newcommand{\dist}{\operatorname{dist}}

\title{Almost Lossless Transfer\\for Cuttable Distributed Problems}
\author{Wojciech Wawrzyniak\\[0.4em]
\normalsize Faculty of Mathematics and Computer Science\\
\normalsize Adam Mickiewicz University\\
\normalsize Uniwersytetu Poznańskiego 4\\
\normalsize 61-614 Poznań, Poland}
\date{11 October 2026}

\begin{document}
\maketitle

\begin{abstract}
\normalsize
We study constant-round LOCAL approximation for additive local
vertex-subset minimization problems. Given a $k$-uniform
$\alpha$-approximation on a hereditary source class closed under
disjoint union, Bonamy, Das,
Gavoille, Picavet, Suomela, and Wesolek obtained ratio
$\alpha(d+1)+1$ on locally nice target classes of asymptotic dimension $d$.
We prove two transfers with ratios $\alpha+1+\eta$ and $\alpha+\eta$
for every $\eta>0$. These bounds require the stronger local niceness
assumption at a larger scale that depends on $\eta$.

The first algorithm keeps their repair construction and increases
the radius used to detect errors. We average over partitions of the graph
to bound the retained cost without the factor $d+1$.
These partitions are used only in the analysis.
The second algorithm tests several deletion widths
around the errors. In each component of the repair zone, it chooses
a smallest union of retained choices and an exact repair.
Averaging over the widths bounds the cost counted twice and yields
ratio $\alpha+\eta$.

Both transfers are deterministic and computable, use arbitrary distinct
numerical identifiers, and need no knowledge of the graph size.
Their outputs also satisfy local uniformity guarantees.
For 2-domination on graphs of bounded Euler genus, the same planar
source gives ratio $19$ through the earlier transfer and $6+\eta$
through ours. For minimum dominating set, our transfer gives ratio
$7+\eps$ on graphs formed by attaching a clique of arbitrary size
at one vertex of a planar graph, using the source algorithm of Amiri
and Siebertz. This class satisfies our target assumptions and allows
nonempty error sets. Its Euler genus is unbounded, so their guarantee
for graphs of bounded Euler genus does not cover the whole class.
\end{abstract}

\clearpage
\section{Introduction}\label{page:main-start}

A local approximation algorithm may work well on one graph class
and still provide no guarantee on a larger class.
We can nevertheless run its local rule on a graph in that larger class.
Three questions then arise. Is the output correct?
How does its cost compare with the optimum?
How can the remaining constraints be repaired in a constant number of rounds?

Minimum dominating set illustrates why such a transfer matters.
A dominating set places a selected vertex at distance at most one from
every vertex of a network. Such sets provide a basic way to select local
coordinators in a network \cite{LPW13}.
On general graphs, a constant number of LOCAL rounds cannot guarantee
a constant approximation ratio \cite{KMW16}.
The following upper bounds with a nonconstant number of rounds use
the identifier conventions of the cited papers, with identifiers
bounded polynomially in $n$.
With a polylogarithmic number of rounds, the randomized approximation
scheme for covering integer programs \cite{GKM17} applies to MDS.
Combined with a deterministic network decomposition
\cite{RG20}, it gives a deterministic approximation
scheme for MDS, also in a polylogarithmic number of rounds.
For planar graphs, a $(1+\eps)$-approximation is possible in
$O(\log^* n)$ rounds, with a matching obstruction to sub-$\log^* n$
deterministic approximation schemes \cite{CHW08}.
The upper bound extends to every fixed excluded-minor class
\cite{CHWW18}.
Thus a constant round bound requires a different trade-off between
locality and approximation.

For constant-round algorithms on planar graphs, the first
approximation ratio was 130 \cite{LPW13}.
A refined analysis of the same algorithm gives 52 \cite{Waw14}.
Later results reduced the ratio to $11+\eps$ \cite{HSV22,HKOdMSV25}.
Wawrzyniak's preprint gives ratio $8+\eps$ \cite{Wawrzyniak26}.
The later preprint of Amiri and Siebertz gives ratio $7+\eps$ \cite{AS26}.
Together with the lower bound $7-\eps$ of \cite{HLS14}, this determines
the infimum of deterministic constant-round approximation ratios
for planar MDS as seven.
Amiri and Siebertz state that Wawrzyniak's result motivated them to complete their work
\cite{AS26}.

The result of Amiri and Siebertz applies to all graphs of Euler genus
at most any fixed integer $g\ge0$, including planar graphs when $g=0$
\cite[Theorem~8.4]{AS26}. For every $\eps>0$, their deterministic LOCAL
algorithm returns a dominating set whose size is at most $7+\eps$ times
the minimum size of a dominating set. Its number of rounds depends only
on $g$ and $\eps$, and it requires no embedding. We use this algorithm
on bounded-Euler-genus graphs as the source for our transfer.

Bonamy and coauthors  address the transfer questions posed at the start
of this introduction through a series of meta-theorems. The statements appear in the conference
version~\cite{BDGPSW26}, and the detailed proofs in~\cite{BDGPSW26full}.
An important contribution of their work is a general transfer method
that applies to many local minimization problems. It is the starting
point for our results.

Our transfer theorems need a cost bound for every region, not only
for the whole graph. We express this requirement as
\textbf{$k$-uniformity}: for any vertex set $W$, the number of vertices
selected in $W$ is at most $\alpha$ times the minimum cost of satisfying
the constraints within distance $k$ of $W$.
For cuttable problems, an ordinary source approximation yields this
regional bound. Cuttability makes this
possible by relating the optimum for regional constraints to the
optimum on a suitable induced subgraph.

The transfer in \cite{BDGPSW26full} assumes that the target class has
\textbf{asymptotic dimension} at most $d$. At each scale $s$, every graph
$G$ in the class admits a colouring with $d+1$ colours. Joining
same-coloured vertices at distance at most $s$ gives components of
diameter at most $f(s)$ in $G$, for one function $f$ shared by the class.
The colouring need not be proper and is used only in the analysis.
Classes excluding a fixed minor have asymptotic dimension at most two
\cite{BBE24}.

Bonamy and coauthors obtain approximation ratio $\alpha(d+1)+1$
\cite{BDGPSW26full}. They also suggested that a finer
analysis of their algorithms could improve the approximation ratio
\cite[Section~6]{BGPW25}.

\paragraph{Our results.}
We prove two transfer theorems.

The algorithm \algname{SingleRepair} keeps the construction of
Meta-Theorem~I and achieves ratio $\alpha+1+\eta$ for every $\eta>0$.
This requires increasing the error-detection radius $\tau$ and assuming
local niceness at the larger scale. The improvement follows from
a tighter analysis of the cost of the retained source output.

Our new algorithm \algname{ShiftedRepair} achieves ratio $\alpha+\eta$
for every $\eta>0$.
For this further improvement, we change the repair step and require
local niceness at a larger scale than in \algname{SingleRepair}.
Within each connected component of the repair zone, the algorithm
tests $m$ buffer widths around the errors.
For each width, it deletes the source output within the buffer and
computes an exact repair. It then chooses the variant with the fewest
vertices in the union of the retained set and the repair.

Both transfers require local niceness at a larger scale than
Meta-Theorem~I of \cite{BDGPSW26}, and this scale depends on $\eta$.
This is a stronger assumption on the target class.
We do not prove these ratios under local niceness only at the original scale.
Section~\ref{sec:limits} gives the exact increase in scale and compares
the uniformity radii and round bounds.

Our contribution is the transfer with arbitrarily small additive loss.
For 2-domination, the planar 6-approximation of Czygrinow and coauthors
\cite{CHSWW17} gives ratio $19$ on bounded-Euler-genus graphs through
\cite[Corollary~4.5]{BDGPSW26full}. Our transfers give $7+\eta$ and
$6+\eta$ on the same graph classes, with the larger scales described
above. Table~\ref{tab:scale-example} compares the resulting bounds.
For MDS,  Section~\ref{sec:seven-source} transfers the source  guarantees
of Amiri and Siebertz  and gives a target class with unbounded Euler
genus
 and  nonempty error  sets.

There is also a useful case in which no repairs are needed.
If every tested ball belongs to the source class, the source algorithm
itself gives a $k$-uniform $(\alpha+\eta/2)$-approximation on the target
class in the same $r$ rounds, by Corollary~\ref{cor:iterate-no-errors}.
The whole target graph may still lie outside the source class.

\section{Model and definitions}\label{sec:model}

We consider finite, simple, undirected graphs. Distances are measured
in the whole graph unless we explicitly specify a different graph.
Vertices in different connected components have distance $\infty$.
For $v\in V(G)$, the \textbf{open neighbourhood} is
$N_G(v)\coloneqq\{u\in V(G):uv\in E(G)\}$.
For $W\subseteq V(G)$ and an integer $a\ge0$, its
\textbf{closed neighbourhood} $N_G^a[W]$ consists of vertices at
distance at most $a$ from some vertex of $W$.
For a single centre, we use the term \textbf{$a$-ball} and write
$N_G^a[v]\coloneqq N_G^a[\{v\}]$.
We omit the subscript $G$ when the graph is clear from the context.
In particular, $\N0W=W$ and $\N a\varnothing=\varnothing$.

The notation $G[W]$ denotes the subgraph of $G$ induced by $W$.
Thus $G[N_G^a[v]]$ is the subgraph induced by the ball $N_G^a[v]$
with centre $v$ and radius $a$.
The \textbf{weak diameter} of a set uses distances in the whole graph
$G$, so shortest paths may leave the set.
The power $G^s$ has the same vertices as $G$ and joins two distinct
vertices if their distance in $G$ is at most $s$.
We use $\mathbb N\coloneqq\{0,1,2,\ldots\}$.

\subsection{The local model}

We use the deterministic \textbf{computable LOCAL model}, with
synchronous communication between neighbours and no bound on message
size or local computation time.
All vertices use the same rule without knowing $n\coloneqq|V(G)|$.
Local rules must terminate on every finite input view, including
outside the source class.

Vertices have distinct identifiers in $\mathbb N$, unrestricted by $n$.
Each edge endpoint has a port number, distinct at each vertex.
Vertex and edge labels are finite strings encoded as natural numbers.
Induced subgraphs keep the original identifiers, labels, and ports.

\begin{definition}\label{def:extended-view}
For $v\in V(G)$ and an integer $a\ge0$, the \textbf{extended $a$-view}
at $v$ is the rooted graph
\[
 \bigl(N_G^{a+1}[v],\,\{e\in E(G):e\cap N_G^a[v]\ne\varnothing\},\,v\bigr),
\]

together with  the identifiers of all vertices in $N_G^{a+1}[v]$,
vertex labels on $N_G^a[v]$,
 and the labels and  endpoint ports of the
displayed edges.

\end{definition}

This information can be collected in $a+1$ rounds.
The extended $r$-view at $v$ determines its output in any deterministic
$r$-round algorithm.

\subsection{The problem and local uniformity}

We fix a \textbf{vertex-subset minimization problem} $\Pi$,
whose solutions are sets of vertices.
For each input graph $G$, the problem specifies which sets
$Y\subseteq V(G)$ are feasible, and the objective is to find such a set
of minimum cardinality. The cost of $Y$ is therefore $|Y|$.

A \textbf{verifier} for $\Pi$ checks a candidate $Y$ and returns
accept or reject at every vertex.
The set $Y$ is feasible exactly when every vertex accepts it.
We require the verifier to terminate on every finite input and write
$F_G(v,Y)$ for the condition that it accepts $Y$ at $v$.

The following notions come from~\cite{BDGPSW26}.
We state them in our model, using a computable verifier and the
extended views defined above.

\begin{definition}\label{def:problem}
For a fixed integer $\rho\ge0$, the problem $\Pi$ is
\textbf{$\rho$-local} if it has a verifier whose answer $F_G(v,Y)$
depends only on the extended $\rho$-view at $v$ and on
$Y\cap\N\rho{v}$.
\end{definition}

Throughout the paper we use one fixed verifier and consider target
instances $G$ with a feasible solution.
The analysis also uses sets that satisfy the constraints only in a
specified region, as follows.

\begin{definition}\label{def:feasible}
For $Y,W\subseteq V(G)$, the set $Y$ is \textbf{feasible on} $W$
if $F_G(v,Y)$ holds for every $v\in W$.
We define
\[
\begin{aligned}
 \OPT_\Pi(G,W)&\coloneqq
 \min\{|Y|:Y\subseteq V(G),\ Y\text{ is feasible on }W\},\\
 \OPT_\Pi(G)&\coloneqq\OPT_\Pi(G,V(G)).
\end{aligned}
\]
\end{definition}

Feasibility on $W$ restricts only the vertices where the constraints
are checked. The selected vertices of $Y$ may also lie outside $W$.
The next property ensures that adding vertices during repair
preserves constraints that have already been satisfied.

\begin{definition}\label{def:monotone}
The problem $\Pi$ is \textbf{additive} if $F_G(v,Y)$, meaning acceptance
of $Y$ at $v$, implies $F_G(v,Y')$ for every $Y'\supseteq Y$.
\end{definition}

To bound the retained cost, we need an approximation guarantee that
holds in every region of the graph.
The following local uniformity condition expresses this requirement.

\begin{definition}\label{def:uniform}
An algorithm $A$ is a \textbf{$k$-uniform $\alpha$-approximation}
on $\C$ if its output is feasible on every instance $G\in\C$
and, for every $W\subseteq V(G)$,
\[
 |A(G)\cap W|\le
 \alpha\,\OPT_\Pi(G,N_G^k[W]).
\]
\end{definition}%

Local uniformity compares the number of choices made by $A$ inside $W$
with the minimum cost of satisfying the constraints in $W$ and in a
buffer of width $k$ around $W$.
It is stronger than the global approximation guarantee, which follows
by setting $W=V(G)$.
Below we assume $\alpha\ge1$ and $k\ge r+1$, where $r$ is the number
of communication rounds of the source algorithm $A$.

The class $\C$ may contain \textbf{labelled instances}, namely graphs
with input labels on vertices or edges, and membership in $\C$
may depend on those labels.
Since our proofs pass to induced subgraphs, we require the following property.

\begin{definition}\label{def:hereditary}
A class of instances $\C$ is \textbf{hereditary} if $G[W]\in\C$
for every $G\in\C$ and $W\subseteq V(G)$, with all input data
inherited from $G$.
\end{definition}

We also assume that membership in $\C$ is decidable.

\subsection{The geometry of the target class}

The cost analysis uses partitions of the graph into bounded-diameter
parts. The following property of the target class enables their construction.

\begin{definition}\label{def:control}
A graph class $\D$ has a \textbf{$d$-dimensional control function}
$f:\mathbb N\to\mathbb N$ if, for every $G\in\D$ and $s\ge1$,
there is a colouring with $d+1$ colours whose monochromatic connected
components in $G^s$ have weak diameter at most $f(s)$ in $G$.
\end{definition}

We assume without loss of generality that $f$ is nondecreasing.
The integer argument $s\ge1$ specifies the scale, and $f(s)$ is a
nonnegative integer bound on weak diameter shared by all graphs in $\D$.
We fix one such function for the class.
For paths, the example in \cite[Section~2]{BDGPSW26full} uses
$f(s)=s-1$ for $s\ge1$.
Divide a path into consecutive blocks of $s$ vertices, allowing the last
block to be shorter, and colour the blocks alternately with two colours.
Each block has weak diameter at most $s-1$ in $G$.
The nearest vertices of distinct blocks of the same colour are at
distance at least $s+1$, so these blocks have no edge between them
in $G^s$. They are therefore the monochromatic connected components
of $G^s$.
It was shown in~\cite{BDGPSW26full} that such a function exists exactly
when the class has asymptotic dimension at most $d$.

Repair may be needed at vertices for which the subgraph induced by the ball
lies outside the source class. These vertices form the error set.

\begin{definition}\label{def:error-set}
For an integer $a\ge0$, the \textbf{$a$-error set} of $G$ relative
to $\C$ is
\[
 X_a(G)\coloneqq\{v\in V(G):G[N_G^a[v]]\notin\C\},
\]
where $G[N_G^a[v]]$ inherits its input data from $G$.
Its elements are called the \textbf{$a$-errors} of $G$.
\end{definition}

When $\C$ is the class of planar graphs, $X_a(G)$ consists of all
vertices $v$ for which $G[N_G^a[v]]$ is nonplanar.
This condition is checked separately at each vertex, even when
the subgraphs under consideration contain the same nonplanar subgraph.
To repair the output in a bounded number of rounds, we need control
over the diameters of regions surrounding the errors.
The following definition serves this purpose.

\begin{definition}\label{def:nice}
A graph $G$ is \textbf{$T$-locally $\delta$-nice} relative to $\C$
and $\Pi$ if every connected component of
$G[\N{2\rho}{X_T(G)}]$ has weak diameter at most $\delta$,
where $\rho$ is the radius of the fixed verifier for $\Pi$.
A class is $T$-locally $\delta$-nice if all its graphs have this property.
\end{definition}

Local niceness lets us collect the data needed for repair in a number
of rounds independent of the graph size.
The $2\rho$-neighbourhoods of error vertices may merge into larger
connected components. The definition ensures that any two vertices
of the same component are at distance at most $\delta$ in $G$.
Each vertex can therefore collect information about the entire component
and the neighbourhood needed to check a repair in a bounded number of rounds.
With the same information, the vertices can compute a common repair.

\section{The transfer algorithm and its guarantee}\label{sec:theorem}

\subsection{Assumptions}\label{sec:guarantees}

We use the following assumptions throughout the rest of the paper.
The problem $\Pi$ is an additive $\rho$-local vertex-subset
minimization problem.
The source class $\C$ is hereditary, and membership in $\C$ is decidable.
For the target class $\D$, we fix a $d$-dimensional control function $f$.

We are also given a source algorithm $A$ for $\Pi$ on $\C$ with the
$k$-uniform $\alpha$-approximation guarantee.
The algorithm is deterministic and computable and runs in $r$ rounds,
where $r\in\mathbb N$. Its local rule terminates on every finite input view.
We assume that its uniformity radius $k$ satisfies $k\ge r+1$.
The approximation ratio $\alpha$ satisfies $\alpha\ge1$.

The source rule, verifier, and procedure for testing membership in $\C$
are fixed in advance. For each transfer algorithm, the corresponding
theorem specifies the required scale of local niceness.
To apply the theorem to a graph class, we need a common integer bound
$\delta\ge0$ on the weak diameters of the components in
Definition~\ref{def:nice} at this scale.
The bound $\delta$ may therefore depend on the class and the scale.
With the problem and source algorithm fixed, we choose $\delta$ and the
radii in advance for the target class and the chosen approximation loss.
These values do not depend on the particular input graph or its identifiers.

The analysis assumes that the input graph satisfies this local niceness
condition. The algorithm detects the vertices in the error set,
but does not have to check the condition globally.

Under these assumptions, Theorem~\ref{thm:A} gives a
$(\delta+k)$-uniform $(\alpha+1+\eta)$-approximation when local niceness
holds at scale $\tau(\eta)$. Theorem~\ref{thm:B} gives a
$(\delta+k)$-uniform $(\alpha+\eta)$-approximation at the larger scale
$T(\eta)$. Both algorithms are deterministic, computable, and take a
constant number of rounds for fixed parameters. We  first  describe
\algname{SingleRepair}  and prove  the cost bound shared by both
algorithms.
 Section~\ref{sec:thmB} then introduces \algname{ShiftedRepair}.

\subsection{\algname{SingleRepair}}\label{sec:overview}

The pseudocode of \algname{SingleRepair} ($\Bsingle$) appears below.

\begin{algorithm}[H]
\caption{Algorithm $\Bsingle$ (\algname{SingleRepair}).}
\label{alg:one-repair}
\begin{algorithmic}[1]
\color{black}
\Require 
An instance $G\in\D$ and the source algorithm $A$, under the assumptions of Theorem~\ref{thm:A}.
\Ensure A feasible set $\Bsingle(G)$ with
$|\Bsingle(G)|\le(\alpha+1+\eta)\OPT_\Pi(G)$.
\State Each vertex $u$ checks whether $G[N_G^\tau[u]]\in\C$.
Vertices for which the test fails form $X$.
\State Run $A$ on the whole graph $G$ for $r$ rounds and set
$S\gets A(G)\setminus X$.
\State Let $R'$ be the vertices where the verifier rejects $S$.
In every connected component $H$ of $G[N_G^{2\rho}[X]]$, choose
the minimum repair $S'_H$ of Definition~\ref{def:outputA}.
Set $S'\gets\bigcup_H S'_H$.
\State Return $\Bsingle(G)\gets S\cup S'$.
\end{algorithmic}
\end{algorithm}

In the first step, each vertex $v$ checks whether the subgraph
$G[N_G^\tau[v]]$ induced by the ball of radius $\tau$ centred at $v$
belongs to $\C$.
Vertices at which the test fails form the error set
$X=X_\tau(G)$.
Thus $X$ contains the centres of balls that fail the test.
For example, for the planar source class, these are vertices $v$ for which
$G[N_G^\tau[v]]$ is nonplanar.

In the second step, we run $A$ on the whole graph $G$.
All vertices, including those in $X$, participate in communication.
Only after computing $A(G)$ do we delete its selected vertices in $X$,
namely $A(G)\cap X$, obtaining $S\coloneqq A(G)\setminus X$.

In the third step, every vertex of $G$ runs the verifier to check
its local constraint for $S$.
We write $R'$ for the vertices $v$ where the verifier finds that
$S$ fails the local constraint of $\Pi$, meaning that $F_G(v,S)$ is false.
For example, for the MDS problem, these are precisely the vertices not dominated by $S$.
Deleting the selected vertices in $X$ may cause local constraints
to fail outside $X$, while constraints at vertices in $X$ may still
be satisfied by $S$. We therefore check the constraints at all vertices
and require the repair to satisfy them at every vertex of $R'$.

The value $\rho$ is the locality radius of the verifier.
For MDS, $\rho=1$, since a selected vertex dominates itself and its neighbours.
Deleting choices in $X$ can therefore violate constraints only
in $\N\rho X$. It is consequently enough to look for repairs in
$\N{2\rho}X$. We use the following notation for the rejection set
and the \textbf{repair zone}.
\begin{equation*}
 R'\coloneqq\{v:\neg F_G(v,S)\},\qquad
 Z'\coloneqq\N{2\rho}X.
\end{equation*}
We keep the notation $R'$ and $Z'$ throughout the analysis of
\algname{SingleRepair}.
In each connected component $H$ of $G[Z']$, we perform a search
and obtain a repair $S'_H$.
By local niceness, any two vertices of the same component are at
distance at most $\delta$ in $G$.
Their distance within the component may be much larger because
shortest paths in $G$ can leave the component.
Communication may use these paths, allowing the information needed
for a common repair to be collected in a bounded number of rounds.

In the fourth step, we combine the repairs into $S'$ and return
$\Bsingle(G)=S\cup S'$.
Since the problem is additive, adding the repairs preserves the
constraints already satisfied by $S$.

The following definition specifies the repairs and the tie-breaking rule.

\begin{definition}\label{def:outputA}
For each connected component $H$ of $G[Z']$, let $S'_H$ be a
minimum-cardinality subset of $\N\rho{R'\cap H}$ feasible on $R'\cap H$.
If several such sets exist, choose the one whose increasing list of
identifiers is first in lexicographic order.
The full repair and the algorithm output are
\begin{equation*}
 S'\coloneqq\bigcup_H S'_H,\qquad \Bsingle(G)\coloneqq S\cup S'.
\end{equation*}
\end{definition}

\subsection{Parameters of \algname{SingleRepair}}\label{sec:theoremA}

\begin{definition}\label{def:parameters}
For $\eta>0$, define
\[
 \begin{gathered}
 q\coloneqq k+\rho+1,\quad
 M\coloneqq\left\lceil\frac{2q(d+1)(2\alpha+\eta)}{\eta}\right\rceil,\quad
 \tau\coloneqq\max\{f(2M)+2M,r+\rho+1\}.
 \end{gathered}
\]
\end{definition}

We use $M$ as the scale of the partition in the cost analysis.
The test for errors uses radius $\tau$, and \algname{SingleRepair} requires
local niceness at the same scale.
We also write $\tau(\eta)$ to show that $\tau$ depends on the chosen loss $\eta$.
We use these definitions of $q$, $M$, and $\tau$ throughout the rest
of the paper unless stated otherwise.

\subsection{Correctness of \algname{SingleRepair}}\label{sec:witness}

In the analysis of \algname{SingleRepair}, we assume that $G\in\D$ is
$\tau$-locally $\delta$-nice relative to $\C$ and $\Pi$.

\subsubsection{Retained output}

The next results restrict the region that requires repair
to the $\rho$-neighbourhood of $X$.

\begin{observation}\label{lem:restrict}
If $Y$ is feasible on $U$, then $Y\cap N_G^\rho[U]$ is feasible on $U$.
\end{observation}
\begin{proof}
For every $v\in U$, we have
$(Y\cap N_G^\rho[U])\cap N_G^\rho[v]=Y\cap N_G^\rho[v]$.
The claim follows directly from Definition~\ref{def:problem}.
\end{proof}

\begin{fact}\label{lem:views}
Let $H$ be an induced subgraph of $G$, and let $W\subseteq V(G)$ satisfy
$N_G^{r+1}[W]\subseteq V(H)$. Then $A$ produces the same output on $W$
in $G$ and in $H$.
\end{fact}
See the argument of \cite[Claim~C.8]{BDGPSW26full}.
The neighbourhood inclusion preserves the extended $r$-view at every
vertex of $W$, including the inherited identifiers, labels, and ports.

\begin{lemma}\label{lem:verifier}
Let $H$ be an induced subgraph of $G$, and let $U,Y\subseteq V(G)$
and $Y_H\subseteq V(H)$ be such that
\[
 N_G^{\rho+1}[U]\subseteq V(H)
 \quad\text{and}\quad
 Y\cap N_G^\rho[U]=Y_H\cap N_G^\rho[U].
\]
Then $Y$ is feasible on $U$ in $G$ if and only if $Y_H$ is feasible on $U$ in $H$.
\end{lemma}
\begin{proof}
For every $u\in U$, the assumption gives
\[
 N_G^{\rho+1}[u]\subseteq N_G^{\rho+1}[U]\subseteq V(H).
\]
Thus $H$ contains all vertices and edges of the extended $\rho$-view at $u$.
Since $H$ is induced and inherits its input data from $G$, this view is
the same in both graphs. Moreover, $N_G^\rho[u]\subseteq N_G^\rho[U]$,
so the assumed equality of intersections implies
\[
 Y\cap N_G^\rho[u]=Y_H\cap N_G^\rho[u].
\]
Thus, for every $u\in U$, the verifier returns the same answer for
$Y$ in $G$ and for $Y_H$ in $H$. This proves the equivalence.
\end{proof}

Fact~\ref{lem:views} and Lemma~\ref{lem:verifier} transfer feasibility from
the subgraphs induced by balls to vertices of $G$ outside $X$.
These are precisely the vertices $v$ for which $G[N_G^\tau[v]]$ belongs
to the source class $\C$. We then use the locality of the verifier to
show that deleting all vertices of $N_G^t[X]$ from $A(G)$ leaves a set
$A(G)\setminus N_G^t[X]$ that is feasible on
$V(G)\setminus N_G^{t+\rho}[X]$.
\algname{SingleRepair} uses only $t=0$. This case occurs in the proof of
\cite[Theorem~4.2]{BDGPSW26full}, where $A(G)\setminus X$ is shown to be
feasible outside $\N\rho X$. We state the lemma for arbitrary $t$
so that it also applies to our second algorithm, \algname{ShiftedRepair},
which has a better approximation ratio.

\begin{lemma}\label{lem:trusted}

Let $G$ be any graph with input data as in Section~\ref{sec:model}.
 For integers $\tau\ge r+\rho+1$ and $t\ge0$, let
\[
 X\coloneqq X_\tau(G),\qquad S(t)\coloneqq A(G)\setminus\N tX.
\]
Then $A(G)$ is feasible on $V(G)\setminus X$, and $S(t)$ is feasible
on $V(G)\setminus\N{t+\rho}X$.
\end{lemma}
\begin{proof}[\color{black}Proof]
The proof has two steps. In Step~1, we establish $F_G(v,A(G))$ for
every $v\notin X$. In Step~2, we establish $F_G(v,S(t))$ for every
$v\notin\N{t+\rho}X$.

\paragraph{Step 1.}
Fix $v\notin X$ and write $G_v\coloneqq G[N_G^\tau[v]]$.
By Definitions~\ref{def:error-set} and~\ref{def:uniform}, we have $G_v\in\C$
and $A(G_v)$ is feasible in $G_v$. In particular, $F_{G_v}(v,A(G_v))$ holds.

To transfer this property to $G$, we verify the assumptions of Lemma~\ref{lem:verifier}
with $H=G_v$, $U=\{v\}$, $Y=A(G)$, and $Y_H=A(G_v)$.
The set $A(G_v)$ is a subset of $V(G_v)$, as required.
The first condition, $N_G^{\rho+1}[v]\subseteq V(G_v)$, follows from
$\tau\ge\rho+1$. The second condition is
$A(G)\cap N_G^\rho[v]=A(G_v)\cap N_G^\rho[v]$.
For this equality, we apply Fact~\ref{lem:views} with $H=G_v$ and
$W=N_G^\rho[v]$. Its neighbourhood condition holds because
\[
 N_G^{r+1}[N_G^\rho[v]]=N_G^{r+\rho+1}[v]
 \subseteq N_G^\tau[v]=V(G_v),
\]
where the inclusion follows from $\tau\ge r+\rho+1$.
Thus $A$ produces the same output on $N_G^\rho[v]$ in
$G$ and $G_v$.
Both conditions of Lemma~\ref{lem:verifier} hold, so
$F_{G_v}(v,A(G_v))$ implies $F_G(v,A(G))$.

\paragraph{Step 2.}
Fix $v\notin\N{t+\rho}X$, so $\dist_G(v,X)>t+\rho$.
Since $v\notin X$, Step~1 gives $F_G(v,A(G))$.
The set $S(t)$ is obtained by deleting the selected vertices in $\N tX$.
None of these vertices lies in $N_G^\rho[v]$.
Hence all selected vertices of $A(G)$ in $N_G^\rho[v]$ remain selected,
and $S(t)\cap N_G^\rho[v]=A(G)\cap N_G^\rho[v]$.
Both sets are considered in the same graph $G$. By
Definition~\ref{def:problem}, the verifier has the same answer for both
sets at $v$. Therefore $F_G(v,S(t))$ holds.
\end{proof}

\subsubsection{Repair}

Recall that $R'=\{v\in V(G):\neg F_G(v,S)\}$ and $Z'=N_G^{2\rho}[X]$.
The set $R'$ contains the vertices at which $S$ fails the local constraint,
and $Z'$ contains the vertices at distance at most $2\rho$ from $X$.
For example, for MDS, $R'$ is the set of vertices not dominated by $S$.

We prove the repair lemma for a general deletion radius $t$ and a general
zone $Z$. Its specialization to $t=0$ gives the properties needed here.
The same statement will later provide the repairs for \algname{ShiftedRepair}.

\Needspace{23\baselineskip}
\begin{lemma}\label{lem:zoneA}\label{lem:repair-exists}
Let $\tau\ge r+\rho+1$, $t\ge0$, $T\ge\tau$, and $\delta\ge0$ be integers.
Suppose that $G\in\D$ is a feasible instance of $\Pi$ and is
$T$-locally $\delta$-nice with respect to $\C$ and $\Pi$.
Let
\[
 X\coloneqq X_\tau(G),\qquad S(t)\coloneqq A(G)\setminus N_G^t[X],
 \qquad R(t)\coloneqq\{v\in V(G):\neg F_G(v,S(t))\}.
\]
Suppose that $Z\subseteq V(G)$ satisfies
\[
 N_G^{t+2\rho}[X]\subseteq Z\subseteq N_G^{2\rho}[X_T(G)].
\]
Then the following properties hold.
\begin{enumerate}[label=\textup{(\alph*)},leftmargin=*]
\item $R(t)\subseteq N_G^{t+\rho}[X]$ and
$N_G^\rho[R(t)]\subseteq N_G^{t+2\rho}[X]\subseteq Z$.
\item Every component of $G[Z]$ has weak diameter at most $\delta$ in $G$.
\item For every component $H$ of $G[Z]$, we have
$N_G^\rho[R(t)\cap H]\subseteq H$, and some subset of
$N_G^\rho[R(t)\cap H]$ is feasible on $R(t)\cap H$.
\end{enumerate}
\end{lemma}
\begin{proof}
For part~(a), Lemma~\ref{lem:trusted} shows that
$S(t)=A(G)\setminus N_G^t[X]$ is feasible on
$V(G)\setminus N_G^{t+\rho}[X]$.
Since $R(t)$ contains only vertices at which $S(t)$ fails,
$R(t)\subseteq N_G^{t+\rho}[X]$.
Every vertex $u\in N_G^\rho[R(t)]$ has a vertex $v\in R(t)$ with
$\dist_G(u,v)\le\rho$. By the inclusion just proved,
$\dist_G(v,X)\le t+\rho$. The triangle inequality gives
\[
 \dist_G(u,X)\le\dist_G(u,v)+\dist_G(v,X)
 \le\rho+(t+\rho)=t+2\rho.
\]
The condition $N_G^{t+2\rho}[X]\subseteq Z$ in the statement of the lemma
means that $Z$ contains every vertex at distance at most $t+2\rho$ from $X$.
The bound just proved therefore implies $u\in Z$.
Since $u$ was arbitrary,
$N_G^\rho[R(t)]\subseteq N_G^{t+2\rho}[X]\subseteq Z$.

For part~(b), consider a component $H$ of $G[Z]$.
The assumption $Z\subseteq N_G^{2\rho}[X_T(G)]$ implies that every path
in $G[Z]$ is also a path in $G[N_G^{2\rho}[X_T(G)]]$.
Thus $H$ lies in one component $K$ of the latter graph.
Since $G$ is $T$-locally $\delta$-nice, Definition~\ref{def:nice}
bounds the weak diameter of $K$ in $G$ by $\delta$.
For any $u,v\in H\subseteq K$, we therefore have $\dist_G(u,v)\le\delta$.
The same bound holds for the weak diameter of $H$.

For part~(c), consider the same component $H$ and a vertex
$x\in N_G^\rho[R(t)\cap H]$.
Then there is a vertex
$v\in R(t)\cap H$ such that $\dist_G(v,x)\le\rho$.
A shortest path from $v$ to $x$ in $G$ has length at most $\rho$.
Every vertex on this path is reached from $v$ by an initial segment
of length at most $\rho$, and hence belongs to $N_G^\rho[v]$.
Since $v\in R(t)$, part~(a) gives
\[
 N_G^\rho[v]\subseteq N_G^\rho[R(t)]\subseteq Z.
\]
The path therefore connects $v\in H$ to $x$ inside $G[Z]$.
Both vertices lie in the same component $H$, so $x\in H$.
Since $x$ was arbitrary, $N_G^\rho[R(t)\cap H]\subseteq H$.

By feasibility of the instance, some set $Y\subseteq V(G)$ is feasible
on $V(G)$. In particular, it is feasible on $R(t)\cap H$.
Observation~\ref{lem:restrict} shows that
$Y\cap N_G^\rho[R(t)\cap H]$ is feasible on $R(t)\cap H$.
The inclusion just proved shows that the candidate repair $Y\cap N_G^\rho[R(t)\cap H]$ is contained in $H$.
Thus it is the required candidate repair. If
$R(t)\cap H=\varnothing$, the empty set is a candidate repair.
\end{proof}

For \algname{SingleRepair}, we apply Lemma~\ref{lem:zoneA} with $t=0$,
$T=\tau$, and $Z=Z'$. Part~(c) guarantees that the repairs $S'_H$
in Definition~\ref{def:outputA} exist.

\subsubsection{Minimum repairs and correctness of the output}

The next fact shows that each repair lies in its component
and that their union is feasible on $R'$.

\begin{fact}\label{fact:minimum-repairA}
Let $\mathcal H$ be the set of components of $G[Z']$.
For $H\in\mathcal H$, let $S'_H$ be a minimum-cardinality subset of
$N_G^\rho[R'\cap H]$ feasible on $R'\cap H$, and let
$S'=\bigcup_{H\in\mathcal H}S'_H$.
Then the following properties hold:
\begin{enumerate}[label=\textup{(\alph*)},leftmargin=*,topsep=3pt,itemsep=2pt]
\item For every $H\in\mathcal H$, we have $S'_H\subseteq H$.
Moreover, $|S'|=\sum_{H\in\mathcal H}|S'_H|$.
\item $S'$ is feasible on $R'$, so $F_G(v,S')$ holds for every $v\in R'$.
\end{enumerate}
\end{fact}
\begin{proof}[\color{black}Proof]
For part~(a), Definition~\ref{def:outputA} and Lemma~\ref{lem:zoneA}(c),
applied as above, give $S'_H\subseteq\N\rho{R'\cap H}\subseteq H$
for every $H\in\mathcal H$. The components are pairwise disjoint, so
the repairs are also pairwise disjoint. By the definition of $S'$,
\[
 |S'|=\left|\bigcup_{H\in\mathcal H}S'_H\right|
 =\sum_{H\in\mathcal H}|S'_H|.
\]

For part~(b), fix $v\in R'$. By Lemma~\ref{lem:zoneA}(a) with $t=0$,
every vertex of $R'$ is at distance at most $\rho$ from $X$.
The repair zone $Z'=N_G^{2\rho}[X]$ consists of all vertices at distance
at most $2\rho$ from $X$. Since $\rho\le2\rho$, we have
\[
 v\in R'\subseteq N_G^\rho[X]\subseteq N_G^{2\rho}[X]=Z'.
\]
Thus $v$ belongs to some component $H$ of $G[Z']$ and belongs to $R'\cap H$.
The repair $S'_H$ from Definition~\ref{def:outputA} is feasible on $R'\cap H$,
so $F_G(v,S'_H)$ holds. Since $S'_H\subseteq S'$, additivity implies $F_G(v,S')$.
This proves part~(b).

\end{proof}

\begin{lemma}\label{lem:feasibleA}
The set $\Bsingle(G)$ returned by $\Bsingle$ on $G$ is feasible.
\end{lemma}%
\begin{proof}
Let $v\in V(G)$. If $v\notin R'$, the definition of $R'$ gives $F_G(v,S)$,
and hence $F_G(v,\Bsingle(G))$.
If $v\in R'$, Fact~\ref{fact:minimum-repairA}(b) gives $F_G(v,S')$,
and hence $F_G(v,\Bsingle(G))$.
Thus $\Bsingle(G)$ is feasible.%
\end{proof}

\section{Approximation ratio of \algname{SingleRepair}}\label{sec:charge}\label{sec:thmA}

In this section, we prove that $\Bsingle$ returns an
$(\alpha+1+\eta)$-approximation. We first choose a colouring and
build a partition of $V(G)$ from the components of each colour.
We use this partition to bound the size of the set $S$ that we keep,
then add the cost of the repair $S'$ to get the required guarantee
for every region $W\subseteq V(G)$.
The partition is used only in the proof, so the algorithm does not compute it.

Fix $G\in\D$.
Throughout this section, distances and weak diameters are measured
in $G$, and balls refer to these distances, unless another graph is specified.
We use $G^{2M}$ to define
the monochromatic components.

\subsection{Constructing the partition}\label{sec:partition}

We first define $D$ in terms of the $d$-dimensional control function $f$
of $\D$ from Definition~\ref{def:control}.
By \cite[Proposition~2.1]{BDGPSW26full}, such a function exists exactly
when the asymptotic dimension of $\D$ is at most $d$.

\begin{definition}\label{def:partition-parameters}
Define
\[
 D\coloneqq f(2M)+2M,\qquad
 \theta\coloneqq\frac{2q(d+1)}{M},\qquad
 \alpha_\theta\coloneqq\frac{\alpha}{1-\theta}.
\]
\end{definition}

\begin{fact}\label{fact:partition-parameters}
We have
\[
 D\le\tau,\qquad 0\le\theta\le\frac{\eta}{2\alpha+\eta}<1.
\]
\end{fact}
\begin{proof}
The definition of $\tau$ gives $\tau\ge f(2M)+2M=D$.
Also, $M\ge2q(d+1)(2\alpha+\eta)/\eta$ by the definition of $M$.
Since $q$, $d+1$, and $M$ are positive, this gives
\[
 0\le\theta=\frac{2q(d+1)}{M}\le\frac{\eta}{2\alpha+\eta}<1,
\]
where the last inequality follows from $\alpha\ge1$ and $\eta>0$.
\end{proof}

As in \cite{BDGPSW26full}, we start the partition construction with
the following colouring.

\begin{definition}\label{def:fixed-colouring}
Fix a colouring $\chi\colon V(G)\to\{0,\ldots,d\}$ given by
Definition~\ref{def:control} with $s=2M$.
The colouring remains fixed throughout the construction.
\end{definition}

\begin{fact}\label{fact:fixed-colouring-distance}
For every monochromatic component $Q$ in $G^{2M}$ and all $x,y\in Q$,
\[
 \dist_G(x,y)\le f(2M).
\]
\end{fact}

\begin{proof}
By the choice of $\chi$ in Definition~\ref{def:fixed-colouring},
the component $Q$ has weak diameter at most $f(2M)$ in $G$.
Thus any two vertices $x,y\in Q$ satisfy the stated bound.
\end{proof}

We then add a step that is not used in the analysis of
\cite{BDGPSW26full}. We modify the partition formed by these components to limit
the number of balls of radius $q$ cut by its parts.
This will allow us to bound the size of $S$.

\begin{definition}\label{def:radius-vector}
For $i\in\{0,\ldots,d\}$ and $v\in V(G)$, define
\[
 C_i\coloneqq\{u\in V(G)\mid\chi(u)=i\},\qquad
 d_i(v)\coloneqq\min_{u\in C_i}\dist_G(v,u),\quad
 \text{where }\min\varnothing=\infty.
\]
A \textbf{radius vector} is any vector
$\mathbf a=(a_0,\ldots,a_d)\in\{1,\ldots,M\}^{d+1}$.
\end{definition}

The sets $C_i$ are the colour classes of the fixed colouring $\chi$.
The value $d_i(v)$ is the distance in $G$ from $v$ to a nearest vertex
of colour $i$, or $\infty$ if no such vertex is reachable from $v$.
Thus $d_i(v)\le a_i$ means that $N_G^{a_i}[v]$ contains a vertex of
colour $i$, or equivalently that $C_i\cap N_G^{a_i}[v]\ne\varnothing$.
We obtain $G^{2M}[C_i]$ by first taking the power of the whole graph $G$
and then taking the subgraph induced by $C_i$.
Two distinct vertices of $C_i$ are adjacent in this graph when their
distance in $G$ is at most $2M$. A path in $G$ realizing this distance
may contain vertices of other colours.

The next fact shows that, when $d_i(v)\le a_i\le M$,
the set $C_i\cap N_G^{a_i}[v]$ lies in one component of $G^{2M}[C_i]$.

\begin{fact}\label{fact:component-unique}
Let $v\in V(G)$, let $i$ be a colour, and let $a_i\le M$ be an integer
such that $d_i(v)\le a_i$. Then there exists a component $Q$ of
$G^{2M}[C_i]$ such that
\[
 C_i\cap N_G^{a_i}[v]\subseteq Q.
\]
\end{fact}
\begin{proof}
For two distinct vertices $x,y\in C_i\cap N_G^{a_i}[v]$, the triangle
inequality gives
\[
 \dist_G(x,y)\le\dist_G(x,v)+\dist_G(v,y)\le2a_i\le2M.
\]
Since $x,y\in C_i$, they are adjacent in $G^{2M}[C_i]$ by the definition
of a graph power. All vertices of $C_i\cap N_G^{a_i}[v]$ therefore
lie in one component of this graph.
\end{proof}

\begin{definition}\label{def:radius-partition}
Let $\mathbf a=(a_0,\ldots,a_d)$ be a radius vector.
Assign to each $v\in V(G)$ the smallest colour
\[
 i\coloneqq\min\{j\in\{0,\ldots,d\}\mid d_j(v)\le a_j\}
\]
and the component $Q$ of $G^{2M}[C_i]$ containing
$C_i\cap N_G^{a_i}[v]$. The \textbf{partition} $\PP_{\mathbf a}$ of $V(G)$
consists of the sets of vertices assigned to the same pair $(i,Q)$.
\end{definition}

A colour satisfying $d_i(v)\le a_i$ always exists, since
$d_{\chi(v)}(v)=0\le a_{\chi(v)}$.
By Fact~\ref{fact:component-unique}, there exists a component $Q$
containing $C_i\cap N_G^{a_i}[v]$. This component is unique because
the set is nonempty and distinct components have disjoint vertex sets.
The chosen colour $i$ may differ from $\chi(v)$.
A nearest vertex of colour $i$ to $v$ in $G$ does not have to be unique,
but the pair $(i,Q)$ is unique. Adjacent vertices in $G$ assigned the
same colour may belong to different parts if they are assigned to
different components of $G^{2M}[C_i]$.
Figure~\ref{fig:partition-pairs} illustrates the parts defined by $(i,Q)$
on a path.
\makeatletter
\begin{figure}[H]
\centering
\color{black}
\let\default@color\current@color

\tikzset{zeromark/.style={draw=black,fill=black},
  onemark/.style={draw=black,fill=white}}

\begin{tikzpicture}[x=1.25cm,y=1cm,color=black,
  every node/.style={font=\small,text=black},
  zero/.style={circle,zeromark,inner sep=2.5pt},
  one/.style={circle,onemark,inner sep=2.5pt}]
 \filldraw[fill=gray!7,draw=gray!65,rounded corners=5pt]
   (-.45,-.65) rectangle (2.55,.85);
 \filldraw[fill=gray!7,draw=gray!65,rounded corners=5pt]
   (2.65,-.65) rectangle (6.55,.85);
 \filldraw[fill=gray!7,draw=gray!65,rounded corners=5pt]
   (6.65,-.65) rectangle (9.65,.85);
 \node at (1.05,.58) {$P_0$};
 \node at (4.6,.58) {$P_1$};
 \node at (8.15,.58) {$P_2$};
 \draw (0,0)--(.95,0);
 \node at (1.3,0) {$\cdots$};
 \draw (1.65,0)--(3.1,0);
 \draw (3.1,0)--(4.2,0);
 \node at (4.6,0) {$\cdots$};
 \draw (5,0)--(7.1,0);
 \draw (7.1,0)--(7.3,0);
 \node at (7.65,0) {$\cdots$};
 \draw (8,0)--(9.2,0);
 \node[zero,label=below:{$v_0$}] at (0,0) {};
 \node[one,label=below:{$v_1$}] at (.65,0) {};
 \node[one,label=below:{$v_M$}] at (2.1,0) {};
 \node[one,label=below:{$v_{M+1}$}] at (3.1,0) {};
 \node[one,label=below:{$v_{3M+1}$}] at (6.1,0) {};
 \node[one,label=below:{$v_{3M+2}$}] at (7.1,0) {};
 \node[one,label=below:{$v_{4M+1}$}] at (8.25,0) {};
 \node[zero,label=below:{$v_{4M+2}$}] at (9.2,0) {};
 \node at (1.05,-1) {$(0,Q_0)$};
 \node at (4.6,-1) {$(1,Q_1)$};
 \node at (8.15,-1) {$(0,Q'_0)$};
 \node[zero] at (1.2,-1.65) {};
 \node[anchor=west] at (1.35,-1.65) {colour $0$};
 \node[one] at (3.9,-1.65) {};
 \node[anchor=west] at (4.05,-1.65) {colour $1$};
 \node at (7.3,-1.65) {$a_0=a_1=M$};
\end{tikzpicture}
\caption{The partition of Definition~\ref{def:radius-partition}.
Vertex colours show the original colouring $\chi$.
The labels $(i,Q)$ identify the partition parts.}
\label{fig:partition-pairs}
\end{figure}
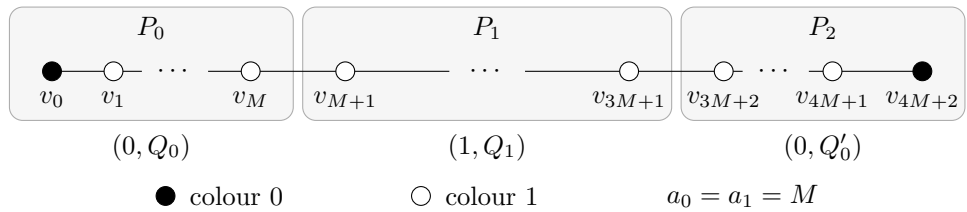
\makeatother

\Needspace{6\baselineskip}
\begin{fact}\label{fact:part-diameter}
Every part of $\PP_{\mathbf a}$ has weak diameter at most $D$ in $G$.
\end{fact}
\begin{proof}
Fix the part $P$ corresponding to $(i,Q)$, where $Q$ is a component
of $G^{2M}[C_i]$. By Definition~\ref{def:radius-partition}, every $v\in P$
has some $w\in Q$ with $\dist_G(v,w)\le a_i\le M$, where
$a_i\le M$ because $\mathbf a$ is a radius vector.
In particular, if $v\in Q$, we may choose $w=v$.
For $x,y\in P$, choose $a,b\in Q$ with
$\dist_G(x,a)\le M$ and $\dist_G(y,b)\le M$.
By Fact~\ref{fact:fixed-colouring-distance}, we have
$\dist_G(a,b)\le f(2M)$, so the triangle inequality in $G$ yields
\[
 \dist_G(x,y)
 \le\dist_G(x,a)+\dist_G(a,b)+\dist_G(b,y)
 \le M+f(2M)+M=D.\qedhere
\]
\end{proof}

The partition in Definition~\ref{def:radius-partition} depends on the
radius vector $\mathbf a$, as indicated by the subscript in
$\PP_{\mathbf a}$. The next fact holds for any fixed radius vector.

\begin{fact}\label{fact:part-in-class}
Let $P\in\PP_{\mathbf a}$ and $u\in P\setminus X$. Then
$P\subseteq N_G^\tau[u]$ and $G[N_G^\tau[u]]\in\C$.
\end{fact}

\begin{proof}
By Fact~\ref{fact:part-diameter}, every vertex of $P$ is at distance
at most $D$ from $u$. Since $D\le\tau$ by
Fact~\ref{fact:partition-parameters}, we have
$P\subseteq N_G^{D}[u]\subseteq N_G^\tau[u]$.
The membership $G[N_G^\tau[u]]\in\C$ follows directly from
Definition~\ref{def:error-set}.
\end{proof}

The next fact considers the case where, for each colour $i$, the
distance $d_i(v)$ is small enough that $d_i(v)\le a_i-q$, large enough
that $d_i(v)\ge a_i+q+1$, or infinite.
These conditions ensure that moving the centre of the ball from $v$
by distance at most $q$ in $G$ does not change whether the ball of
radius $a_i$ contains a vertex of colour $i$.

\begin{fact}\label{fact:uncut-ball}
Let $v\in V(G)$ and suppose that $\mathbf a\in\{1,\ldots,M\}^{d+1}$
satisfies, for every colour $i$,
\[
 d_i(v)=\infty
 \quad\text{or}\quad d_i(v)\le a_i-q
 \quad\text{or}\quad d_i(v)\ge a_i+q+1.
\]
Then the whole ball $N_G^q[v]$ lies in one part of $\PP_{\mathbf a}$.
\end{fact}

\begin{proof}
Let $z\in N_G^q[v]$. We first check that $z$ accepts the same colours
as $v$. That is, $d_i(z)\le a_i$ holds exactly for those colours $i$
for which $d_i(v)\le a_i$.
For finite distances, the triangle inequality in $G$ gives
$|d_i(v)-d_i(z)|\le\dist_G(v,z)\le q$.
If $d_i(v)\le a_i-q$, then $d_i(z)\le d_i(v)+q\le a_i$,
so both $v$ and $z$ accept $i$.
If $d_i(v)\ge a_i+q+1$, then $d_i(z)\ge d_i(v)-q\ge a_i+1$,
so neither accepts $i$.
If $d_i(v)=\infty$, the component of $G$ containing $v$ has no vertex
of colour $i$. Since $z\in N_G^q[v]$ lies in the same component,
$d_i(z)=\infty$ as well.
Thus $v$ and $z$ have the same smallest colour $i$ satisfying the
condition of Definition~\ref{def:radius-partition}.
For this colour, $d_i(v)\le a_i-q$.

It remains to check that the vertices $v$ and $z$ are assigned to the same component
of $G^{2M}[C_i]$.
For the colour $i$ assigned to both vertices, we have already shown
that $d_i(v)\le a_i-q<\infty$.
Since $d_i(v)$ is the minimum distance from $v$ to $C_i$, we can choose
a vertex $w\in C_i$ with $\dist_G(v,w)=d_i(v)\le a_i-q$.
Also, $z\in N_G^q[v]$ implies $\dist_G(z,v)\le q$.
The triangle inequality in $G$ now gives
\[
 \dist_G(z,w)\le\dist_G(z,v)+\dist_G(v,w)
 \le q+(a_i-q)=a_i.
\]
Thus $w$ belongs to both $C_i\cap N_G^{a_i}[v]$ and
$C_i\cap N_G^{a_i}[z]$.
By Fact~\ref{fact:component-unique}, the components assigned to $v$
and $z$ both contain $w$, so they are equal.
Both vertices are therefore assigned to the same pair $(i,Q)$.
Since $z\in N_G^q[v]$ was arbitrary, the whole ball lies in the part
of $\PP_{\mathbf a}$ containing $v$.
\end{proof}

\begin{definition}\label{def:cut-ball}
We say that $\PP_{\mathbf a}$ \textbf{cuts} $N_G^q[v]$ if there exist
vertices $u,w\in N_G^q[v]$ and distinct parts $P_1,P_2\in\PP_{\mathbf a}$
such that $u\in P_1$ and $w\in P_2$.
\end{definition}

By Fact~\ref{fact:uncut-ball}, a cut is possible only when, for some colour $i$,
$a_i\in[d_i(v)-q,d_i(v)+q-1]$.
This interval contains only $2q$ integers, which gives the following count.

\begin{fact}\label{fact:cut-count}
For every $v\in V(G)$, at most $2q(d+1)M^d$ vectors
$\mathbf a\in\{1,\ldots,M\}^{d+1}$ have the property that
$\PP_{\mathbf a}$ cuts $N_G^q[v]$.
\end{fact}
\begin{proof}
If $\PP_{\mathbf a}$ cuts $N_G^q[v]$, the condition of
Fact~\ref{fact:uncut-ball} fails for some colour $i$.
Hence $a_i\in[d_i(v)-q,d_i(v)+q-1]$.
For a fixed colour, there are at most $2q$ such values of $a_i$.
The remaining $d$ coordinates are arbitrary, so this accounts for
at most $2qM^d$ vectors. Summing over the $d+1$ colours gives the bound.
\end{proof}

We now choose the radii so that, for at least $(1-\theta)|J|$ vertices
$v\in J$, the whole ball $N_G^q[v]$ lies in one part of the partition.
For $J\subseteq S$, we will use this property to apply the source algorithm's
guarantee locally and bound $|J|$.

\begin{lemma}\label{lem:padding}
For every $J\subseteq V(G)$, there are a partition $\PP$ of $V(G)$
into sets of weak diameter at most $D$ in $G$ and a set $I\subseteq J$
such that
\[
 |I|\ge(1-\theta)|J|
 \qquad\text{and}\qquad
 \forall v\in I\ \exists P\in\PP:\quad N_G^q[v]\subseteq P.
\]
Moreover, for every $P\in\PP$ and $u\in P\setminus X$,
\[
 P\subseteq N_G^\tau[u],\qquad G[N_G^\tau[u]]\in\C.
\]
\end{lemma}
\begin{proof}
By Fact~\ref{fact:cut-count}, there are at most $2q(d+1)M^d|J|$
pairs $(v,\mathbf a)$ with $v\in J$ for which $\PP_{\mathbf a}$ cuts
$N_G^q[v]$.
There are $M^{d+1}$ radius vectors, so the average number of cut balls
centred in $J$ per vector is at most
$2q(d+1)M^d|J|/M^{d+1}=\theta|J|$.%

Consequently, some vector $\mathbf a$ cuts at most $\theta|J|$ balls
centred in $J$. Set $\PP\coloneqq\PP_{\mathbf a}$, and let $I$ consist
of the vertices $v\in J$ whose balls $N_G^q[v]$ are not cut.
Then $|I|\ge|J|-\theta|J|=(1-\theta)|J|$, and for every $v\in I$
there is a part $P\in\PP$ such that $N_G^q[v]\subseteq P$.
By Fact~\ref{fact:part-diameter}, each $P\in\PP$ has weak diameter
at most $D$ in $G$. Fact~\ref{fact:part-in-class} gives
$P\subseteq N_G^\tau[u]$ and $G[N_G^\tau[u]]\in\C$
for $u\in P\setminus X$.
\end{proof}

\subsection{Bounding the size of the retained set}\label{sec:retained-cost}

Our aim is to bound the size of any $J\subseteq S=A(G)\setminus X$ by
$\alpha_\theta|Y\cap N_G^q[J]|$, where $Y\subseteq V(G)$ is feasible
on $N_G^k[J]$ in $G$.
Recall that the $k$-uniform guarantee means that, for every $H\in\C$
and $W\subseteq V(H)$, $A$ chooses at most $\alpha$ times as many vertices
in $W$ as are needed to satisfy the constraints on $N_H^k[W]$.
Thus $k$ specifies how far beyond $W$ the constraints must be satisfied.
We use $q=k+\rho+1$, where $\rho$ is the verifier's radius.

\begin{fact}\label{fact:alpha-theta}
The coefficient $\alpha_\theta=\alpha/(1-\theta)$ satisfies
$1\le\alpha_\theta\le\alpha+\eta/2$.
\end{fact}%

\begin{proof}
Fact~\ref{fact:partition-parameters} gives $0\le\theta<1$ and
$1-\theta\ge2\alpha/(2\alpha+\eta)$.
Thus $\alpha_\theta\ge\alpha\ge1$, and
\[
 \alpha_\theta
 =\frac{\alpha}{1-\theta}
 \le\frac{\alpha}{\frac{2\alpha}{2\alpha+\eta}}
 =\frac{\alpha(2\alpha+\eta)}{2\alpha}
 =\frac{2\alpha+\eta}{2}
 =\alpha+\frac{\eta}{2}.
\]
\end{proof}%

\begin{lemma}\label{lem:charge}
If $J\subseteq S$ and $Y\subseteq V(G)$ is feasible on $N_G^k[J]$
in $G$, then
\[
 |J|\le\alpha_\theta\,|Y\cap N_G^q[J]|.
\]
\end{lemma}%
\begin{proof}[\color{black}Proof]
Choose a partition $\PP$ and a set $I\subseteq J$ by applying
Lemma~\ref{lem:padding} to $J$.
In particular, $(1-\theta)|J|\le|I|$.
For every $P\in\PP$, write $I_P\coloneqq I\cap P$ and
$Y_P\coloneqq Y\cap N_G^q[I_P]$.
If $I_P\ne\varnothing$, choose $u_P\in I_P$ and let
$G_P\coloneqq G[N_G^\tau[u_P]]$.
The following three claims establish $|I_P|\le\alpha|Y_P|$ for every
$P\in\PP$.%

\begin{proofclaim}\label{claim:charge-subgraph}
For every $P\in\PP$, if $I_P\ne\varnothing$, then
\[
 G_P\in\C\qquad\text{and}\qquad
 N_G^q[I_P]\subseteq P\subseteq V(G_P).
\]
\end{proofclaim}%
\begin{proof}[\color{black}Proof of Claim~\ref{claim:charge-subgraph}]
Choose any vertex $v\in I_P$. Since $I_P=I\cap P$, we have both
$v\in I$ and $v\in P$. We chose $\PP$ and $I$ using
Lemma~\ref{lem:padding}, so the ball $N_G^q[v]$ lies in one part of $\PP$.
This ball contains $v$, and the parts are disjoint, so that part is $P$.
This holds for every $v\in I_P$, hence $N_G^q[I_P]\subseteq P$.%

Now consider the vertex $u_P\in I_P$ chosen at the start of the proof of
Lemma~\ref{lem:charge}. It belongs to $P$, and moreover
$I_P\subseteq J\subseteq S=A(G)\setminus X$ shows that $u_P\notin X$.
Thus $u_P\in P\setminus X$, so we can apply the last part of
Lemma~\ref{lem:padding}. Using the definition of $G_P$, we obtain
\[
 P\subseteq N_G^\tau[u_P]=V(G_P)
 \qquad\text{and}\qquad
 G_P=G[N_G^\tau[u_P]]\in\C.
\]
\end{proof}%

\begin{proofclaim}\label{claim:charge-local-transfer}
For every $P\in\PP$ with $I_P\ne\varnothing$,
\begin{enumerate}[label=(\alph*),leftmargin=*]
\item $I_P\subseteq A(G_P)$.
\item $Y_P$ is feasible on $N_{G_P}^k[I_P]$ in $G_P$.
\end{enumerate}
\end{proofclaim}%
\begin{proof}[\color{black}Proof of Claim~\ref{claim:charge-local-transfer}]
By the assumption $k\ge r+1$ in Section~\ref{sec:guarantees} and
$\rho\ge0$, the definition of $q$ gives $r+1\le k\le k+\rho+1=q$.
Hence
\[
 N_G^{r+1}[I_P]
 \overset{\text{Def.~\ref{def:parameters}}}{\subseteq}
 N_G^q[I_P]
 \overset{\text{Claim~\ref{claim:charge-subgraph}}}{\subseteq}V(G_P).
\]
Fact~\ref{lem:views} therefore gives
$A(G_P)\cap I_P=A(G)\cap I_P$.
Since $I_P\subseteq J\subseteq S\subseteq A(G)$,
we obtain $I_P\subseteq A(G_P)$.%

For part~(b), fix any $z\in N_{G_P}^k[I_P]$.
Every path in $G_P$ is also a path in $G$, so
\[
 z\in N_{G_P}^k[I_P]\subseteq N_G^k[I_P]\subseteq N_G^k[J].
\]
The assumption on $Y$ gives $F_G(z,Y)$.
For every $x\in N_G^{\rho+1}[z]$, the triangle inequality in $G$ implies
\[
 \dist_G(x,I_P)\le\dist_G(x,z)+\dist_G(z,I_P)
 \le\rho+1+k=q.
\]%
Together with Claim~\ref{claim:charge-subgraph}, this gives
\[
 N_G^{\rho+1}[z]\subseteq N_G^q[I_P]\subseteq V(G_P).
\]
Thus $G_P$ contains the extended $\rho$-view at $z$.
Recall that $Y_P=Y\cap N_G^q[I_P]$, so the inclusion
$N_G^\rho[z]\subseteq N_G^q[I_P]$ implies
\[
 Y_P\cap N_G^\rho[z]=Y\cap N_G^\rho[z].
\]%
This equality means that $Y$ and $Y_P$ select the same vertices
in $N_G^\rho[z]$.
Since $G_P$ contains the extended $\rho$-view at $z$ and
$Y_P\subseteq N_G^q[I_P]\subseteq V(G_P)$,
Lemma~\ref{lem:verifier} implies that the verifier gives the same answer
at $z$ for $Y$ in $G$ and for $Y_P$ in $G_P$. Thus
\[
 F_G(z,Y)\quad\Longleftrightarrow\quad F_{G_P}(z,Y_P).
\]
We have already proved $F_G(z,Y)$, so $F_{G_P}(z,Y_P)$ holds.
Since $z$ was arbitrary in $N_{G_P}^k[I_P]$,
Definition~\ref{def:feasible} now implies that $Y_P$ is feasible on
$N_{G_P}^k[I_P]$ in $G_P$.
\end{proof}%

\begin{proofclaim}\label{claim:charge-part}
For every $P\in\PP$, we have $|I_P|\le\alpha|Y_P|$.
\end{proofclaim}%

\begin{proof}[\color{black}Proof of Claim~\ref{claim:charge-part}]
Fix $P\in\PP$. If $I_P=\varnothing$, the inequality holds immediately.
Otherwise, Claim~\ref{claim:charge-subgraph} gives $G_P\in\C$ and
$N_G^q[I_P]\subseteq V(G_P)$.
Apply the $k$-uniform $\alpha$-approximation guarantee of $A$ from
Definition~\ref{def:uniform} with $W=I_P$. This gives
\[
 |A(G_P)\cap I_P|\le\alpha\OPT_\Pi(G_P,N_{G_P}^k[I_P]).
\]
By Claim~\ref{claim:charge-local-transfer}(a),
\[
 |I_P|\le\alpha\OPT_\Pi(G_P,N_{G_P}^k[I_P]).
\]
Part~(b) of that claim shows that $Y_P$ is feasible on $N_{G_P}^k[I_P]$
in $G_P$. By Definition~\ref{def:feasible}, this optimum is the minimum
size of a subset of $V(G_P)$ that is feasible on $N_{G_P}^k[I_P]$ in $G_P$.
Since $Y_P$ is one of these sets, we obtain
$\OPT_\Pi(G_P,N_{G_P}^k[I_P])\le|Y_P|$.
Combining these bounds proves $|I_P|\le\alpha|Y_P|$.
\end{proof}%

We now complete the proof of Lemma~\ref{lem:charge} by combining the
bounds from Claim~\ref{claim:charge-part} over all parts of $\PP$.
The parts of $\PP$ are disjoint and $Y_P\subseteq P$, so the sets $Y_P$
are pairwise disjoint. Moreover, $I_P\subseteq J$, so every vertex at
distance at most $q$ from $I_P$ is also at distance at most $q$ from $J$.
Thus $N_G^q[I_P]\subseteq N_G^q[J]$, and the definition of $Y_P$ gives
\[
 Y_P=Y\cap N_G^q[I_P]\subseteq Y\cap N_G^q[J].
\]
Therefore
\[
 \sum_{P\in\PP}|Y_P|
 =\left|\bigcup_{P\in\PP}Y_P\right|
 \le|Y\cap N_G^q[J]|.
\]%
Every vertex of $I$ belongs to exactly one set $I_P=I\cap P$, so
$|I|=\sum_{P\in\PP}|I_P|$.
Using Claim~\ref{claim:charge-part} and the preceding bound, we obtain
\[
 |I|=\sum_{P\in\PP}|I_P|
 \le\alpha\sum_{P\in\PP}|Y_P|
 \le\alpha|Y\cap N_G^q[J]|.
\]
Lemma~\ref{lem:padding} gives $(1-\theta)|J|\le|I|$, and
Fact~\ref{fact:partition-parameters} gives $1-\theta>0$.
Dividing by $1-\theta$ and using the definition of $\alpha_\theta$ in
Definition~\ref{def:partition-parameters}, we obtain
\[
 |J|\le\frac{\alpha}{1-\theta}|Y\cap N_G^q[J]|
 =\alpha_\theta|Y\cap N_G^q[J]|.
\]
This proves Lemma~\ref{lem:charge}.
\end{proof}%

\subsection{Total cost}\label{sec:repair-budget}%

The set $S=A(G)\setminus X$ is obtained from the output of the source
algorithm $A$ by removing all vertices in $X$.
Lemma~\ref{lem:charge} bounds the size of any subset of $S$.
We combine this bound with the cost of repairs. Recall that $Z'=N_G^{2\rho}[X]$ is the repair
zone and $G[Z']$ is the subgraph it induces.
In each component $H$ of this subgraph, the algorithm chooses a repair
$S'_H$ according to Definition~\ref{def:outputA}.
We first consider a region $U$ such that every component of $G[Z']$
is either contained in $U$ or has no vertex in $U$. Distances and feasibility in the following bound refer to $G$.%

\begin{lemma}\label{lem:budgetA}
Let $U\subseteq V(G)$ satisfy $H\subseteq U$ or $H\cap U=\varnothing$
for every component $H$ of $G[Z']$.
Let $Y\subseteq V(G)$ be feasible on $N_G^k[U]$ in $G$. Then
\[
 |\Bsingle(G)\cap U|
 \le\alpha_\theta\,|Y\cap N_G^q[U]|+|Y\cap U|
 \le(\alpha_\theta+1)|Y|.
\]
\end{lemma}%
\begin{proof}[\color{black}Proof]
Let $\mathcal H$ be the set of components of $G[Z']$.
By Definition~\ref{def:outputA},
$\Bsingle(G)=S\cup\bigcup_{H\in\mathcal H}S'_H$.
Intersecting both sides with $U$, we obtain
\[
 \Bsingle(G)\cap U
 =(S\cap U)\cup\bigcup_{H\in\mathcal H}(S'_H\cap U).
\]
The union above runs over all components $H\in\mathcal H$.
By the assumption on $U$, each such component is either contained in $U$
or has no vertex in $U$.
Fact~\ref{fact:minimum-repairA}(a) gives $S'_H\subseteq H$.
If $H\subseteq U$, then the whole repair $S'_H$ lies in $U$,
so $S'_H\cap U=S'_H$.
If $H\cap U=\varnothing$, then $S'_H$ has no vertex in $U$,
so $S'_H\cap U=\varnothing$.
Thus only repairs in components contained in $U$ remain in the union, and
\[
 \Bsingle(G)\cap U
 =(S\cap U)\cup\bigcup_{\substack{H\in\mathcal H\\H\subseteq U}}S'_H.
\]
The size of a union is at most the sum of the sizes, so
\[
 |\Bsingle(G)\cap U|\le|S\cap U|
 +\sum_{\substack{H\in\mathcal H\\H\subseteq U}}|S'_H|.
\]
We bound the two terms separately.%

By the definition of $S$, we have $S\cap U\subseteq S=A(G)\setminus X$.
Also, $S\cap U\subseteq U$, so $N_G^k[S\cap U]\subseteq N_G^k[U]$.
The assumption of this lemma states that $Y$ is feasible on $N_G^k[U]$.
Thus the local constraint holds at every vertex of $N_G^k[U]$,
and hence at every vertex of its subset $N_G^k[S\cap U]$.
By Definition~\ref{def:feasible}, $Y$ is therefore feasible on $N_G^k[S\cap U]$.
Applying Lemma~\ref{lem:charge} with $J=S\cap U$ yields
$|S\cap U|\le\alpha_\theta|Y\cap N_G^q[S\cap U]|
\le\alpha_\theta|Y\cap N_G^q[U]|$.%

We now bound the total cost of repairs in $U$.
Fix any component $H\in\mathcal H$ with $H\subseteq U$.
Recall that $R'=\{v\in V(G):\neg F_G(v,S)\}$ is the set of vertices
where the verifier rejects $S$.
Since $R'\cap H\subseteq U$, the set $Y$ is feasible on $R'\cap H$.
By Observation~\ref{lem:restrict}, $Y\cap N_G^\rho[R'\cap H]$ is a candidate
for $S'_H$, so $|S'_H|\le|Y\cap N_G^\rho[R'\cap H]|$.
Lemma~\ref{lem:zoneA}(c), with $t=0$, $T=\tau$, and $Z=Z'$,
gives $N_G^\rho[R'\cap H]\subseteq H$.
These candidates therefore lie in disjoint components of $G[Z']$
contained in $U$. The sum of their sizes is at most $|Y\cap U|$.%

Combining the bounds on the retained set and the repairs gives
\[
 |\Bsingle(G)\cap U|
 \le\alpha_\theta|Y\cap N_G^q[U]|+|Y\cap U|.
\]
Both intersections are subsets of $Y$, so
\[
 \alpha_\theta|Y\cap N_G^q[U]|+|Y\cap U|
 \le\alpha_\theta|Y|+|Y|=(\alpha_\theta+1)|Y|.
\]
\end{proof}%
\subsection{The guarantee for an arbitrary region}\label{sec:uniform-guarantee}

The preceding lemmas show that $\Bsingle$ returns an
$(\alpha+1+\eta)$-approximation. To use $\Bsingle$ as the source algorithm
in another transfer, we also need a uniform approximation guarantee.
We now prove that it is a $(\delta+k)$-uniform
$(\alpha+1+\eta)$-approximation. Section~\ref{sec:iteration} describes
how such a guarantee can be used in a further transfer without errors.

Uniformity means that, for every $W\subseteq V(G)$, the number of
vertices in $\Bsingle(G)\cap W$ is at most $(\alpha+1+\eta)$ times the
minimum size of a set feasible on $N_G^{\delta+k}[W]$ in $G$.
The following lemma isolates the extension from regions containing
whole components to arbitrary regions. Both algorithms use this argument.

\begin{lemma}\label{lem:uniform-extension}
Let $G$ be an instance for which $\Pi$ has a feasible solution.
Let $B,Z\subseteq V(G)$, $c\ge0$, and $\delta,k\in\mathbb N$
satisfy the following two conditions.
\begin{enumerate}[label=\textup{(\alph*)},leftmargin=*]
\item Every component of $G[Z]$ has weak diameter at most $\delta$ in $G$.
\item For every $U\subseteq V(G)$ such that each component of $G[Z]$
is either contained in $U$ or has no vertex in $U$, we assume that
\[
 |B\cap U|\le c\,\OPT_\Pi(G,N_G^k[U]).
\]
\end{enumerate}
Then, for every $W\subseteq V(G)$,
\[
 |B\cap W|\le c\,\OPT_\Pi(G,N_G^{\delta+k}[W]).
\]
\end{lemma}
\begin{proof}[\color{black}Proof]
Fix $W\subseteq V(G)$, and let $\mathcal H$ be the family of vertex sets
of the components of $G[Z]$.
Define
\[
 U\coloneqq(W\setminus Z)\cup
 \bigcup_{\substack{H\in\mathcal H\\H\cap W\ne\varnothing}}H.
\]
Fix any $v\in W$. If $v\notin Z$, then $v\in W\setminus Z\subseteq U$.
If $v\in Z$, let $H\in\mathcal H$ be the component of $G[Z]$ containing $v$.
Since $v\in H\cap W$, this intersection is nonempty, so we include all of $H$
in $U$. Hence $v\in H\subseteq U$.
Thus every vertex of $W$ lies in $U$, so $W\subseteq U$.
Since the components are pairwise disjoint and contained in $Z$,
each $H\in\mathcal H$ satisfies $H\subseteq U$ or $H\cap U=\varnothing$.
By (a), every component $H$ added to $U$ has weak diameter at most $\delta$ in $G$.
We add $H$ only if $H\cap W\ne\varnothing$, so $H$ contains a vertex of $W$.
Thus each of its vertices is at distance
at most $\delta$ from $W$, and
\[
 U\subseteq N_G^\delta[W],\qquad
 N_G^k[U]\subseteq N_G^{\delta+k}[W].
\]
The second inclusion follows by taking $k$-neighbourhoods of both sides of the first
inclusion and using $N_G^k[N_G^\delta[W]]=N_G^{\delta+k}[W]$.

Choose a set $Y\subseteq V(G)$ of minimum size that is feasible on $N_G^{\delta+k}[W]$.
By Definition~\ref{def:feasible}, $|Y|=\OPT_\Pi(G,N_G^{\delta+k}[W])$.
Since $N_G^k[U]\subseteq N_G^{\delta+k}[W]$, the set $Y$ is also feasible on $N_G^k[U]$.
Thus a set of size $|Y|$ is feasible on $N_G^k[U]$.
The minimum size of a set feasible on this region is therefore at most $|Y|$, so
\[
 \OPT_\Pi(G,N_G^k[U])\le |Y|=\OPT_\Pi(G,N_G^{\delta+k}[W]).
\]
Using $W\subseteq U$, condition (b), and $c\ge0$, we obtain
\[
 |B\cap W|\le|B\cap U|\le c\,\OPT_\Pi(G,N_G^k[U])
 \le c\,\OPT_\Pi(G,N_G^{\delta+k}[W]).
\]
The argument includes $W=\varnothing$, when $U=\varnothing$.
\end{proof}

\begin{lemma}\label{lem:uniformA}
For every $W\subseteq V(G)$, the output $\Bsingle(G)$ satisfies
\[
 |\Bsingle(G)\cap W|
 \le(\alpha+1+\eta)\OPT_\Pi(G,N_G^{\delta+k}[W]).
\]
\end{lemma}%

\begin{proof}
We first check the assumptions of Lemma~\ref{lem:uniform-extension} for
$B=\Bsingle(G)$, $Z=Z'=N_G^{2\rho}[X]$, and $c=\alpha_\theta+1$.
Applying Lemma~\ref{lem:zoneA}(b) with $t=0$, $T=\tau$, and $Z=Z'$
shows that every component of $G[Z']$ has weak diameter at most $\delta$ in $G$.
This verifies condition (a).

To check condition (b), fix any $U\subseteq V(G)$ such that each component of $G[Z']$
is either contained in $U$ or has no vertex in $U$.
Choose a set $Y\subseteq V(G)$ of minimum size that is feasible on $N_G^k[U]$.
By Definition~\ref{def:feasible}, $|Y|=\OPT_\Pi(G,N_G^k[U])$.
The condition on $U$ and the feasibility of $Y$ allow us to apply Lemma~\ref{lem:budgetA}.
We obtain
\[
 |\Bsingle(G)\cap U|\le(\alpha_\theta+1)|Y|
 = (\alpha_\theta+1)\OPT_\Pi(G,N_G^k[U]).
\]
Since $B=\Bsingle(G)$ and $c=\alpha_\theta+1$, this is the inequality in
condition (b) of Lemma~\ref{lem:uniform-extension}.
The choice of $U$ was arbitrary, so condition (b) of Lemma~\ref{lem:uniform-extension}
holds for every such $U$.

Applying Lemma~\ref{lem:uniform-extension} now gives the first inequality below.
\[
 \begin{aligned}
 |\Bsingle(G)\cap W|
 &\overset{\mathclap{\text{Lem.~\ref{lem:uniform-extension}}}}{\le}
 (\alpha_\theta+1)\OPT_\Pi(G,N_G^{\delta+k}[W])\\
 &\overset{\mathclap{\text{Fact~\ref{fact:alpha-theta}}}}{\le}
 (\alpha+1+\eta/2)\OPT_\Pi(G,N_G^{\delta+k}[W])\\
 &\overset{\mathclap{\eta>0}}{\le}
 (\alpha+1+\eta)\OPT_\Pi(G,N_G^{\delta+k}[W]).
 \end{aligned}
\]
The second inequality uses Fact~\ref{fact:alpha-theta}, and the last
uses $\eta>0$.
\end{proof}%

\section{Analysis of \algname{ShiftedRepair}}\label{sec:thmB}%

In the cost analysis of \algname{SingleRepair}, we bounded the size of
$S=A(G)\setminus X$ and the total cost of repairs separately using the same set $Y$.
To obtain ratio $\alpha+\eta$, \algname{ShiftedRepair} considers buffers of
different widths around the error set $X$.
For each width, it deletes the source choices in the buffer and computes
a repair. It then chooses a smallest variant in each component of the
repair zone. The spacing between the buffer widths bounds the cost
of the part of $Y$ counted in both comparisons.%

We first fix the parameters and describe the algorithm.
We then prove the existence of repairs and feasibility of the output,
bound its cost, and establish the uniform guarantee.%

\subsection{Parameters}%

The error set $X$ is still determined using balls of radius $\tau$.
To bound the weak diameters of the components of the enlarged repair
zone, we require local niceness at a larger scale $T$.
The parameters below satisfy
\[
 T=\tau+\lceil2/\eta\rceil(q+2\rho),\qquad
 \tau<T<\frac74\tau.
\]
Thus, for the same $\eta$, the radius increases by less than $75\%$.%

\begin{definition}\label{def:shift-parameters}
For $\eta>0$, let
\begin{equation*}
 m\coloneqq\lceil2/\eta\rceil,\qquad
 w\coloneqq q+2\rho,\qquad
 L\coloneqq mw,\qquad
 T\coloneqq\tau+L.
\end{equation*}
For each $j\in\{0,\ldots,m-1\}$, the deletion radius is
\[
 t_j\coloneqq jw+q.
\]
\end{definition}%

The number $m$ specifies how many deletion widths the algorithm considers,
and $w$ is the difference between consecutive radii $t_j$.
All repairs are computed within a zone of radius $L$ around $X$.
We also write $L(\eta)$ and $T(\eta)$ to make their dependence on the
chosen loss explicit.%

The next fact records the parameter relations used to locate repairs
and to bound the final cost.
We retain $\alpha_\theta=\alpha/(1-\theta)$ from
Definition~\ref{def:partition-parameters}.%

\begin{fact}\label{fact:shift-parameters}
The parameters of Definition~\ref{def:shift-parameters} satisfy the following.
\begin{enumerate}[label=\textup{(\alph*)},leftmargin=*]
\item $m\ge1$ and, for every $j\in\{0,\ldots,m-1\}$,
\[
 t_j+2\rho=(j+1)w\le L,\qquad
 \max_{0\le j<m}t_j=t_{m-1}=L-2\rho.
\]
\item $1/m\le\eta/2$, $1\le\alpha_\theta$,
$\alpha_\theta+1/m\le\alpha+\eta$.
\item $0<L<\frac34\tau$, $\tau<T<\frac74\tau$.
\end{enumerate}
\end{fact}%

\begin{proof}
Since $\eta>0$, we have $m=\lceil2/\eta\rceil\ge1$.
By the definitions of $w$ and $t_j$,
\[
 t_j+2\rho=jw+q+2\rho=(j+1)w\le mw=L.
\]
The radii increase with $j$ because $w=q+2\rho>0$.
The largest is $t_{m-1}=(m-1)w+q=L-2\rho$, proving part~(a).

The bound $m\ge2/\eta$ and Fact~\ref{fact:alpha-theta} give
$1/m\le\eta/2$ and $1\le\alpha_\theta\le\alpha+\eta/2$.
Thus $\alpha_\theta+1/m\le\alpha+\eta$, proving part~(b).

For part~(c), recall that $q=k+\rho+1>\rho$, so $w=q+2\rho<3q$.
Also, $m=\lceil2/\eta\rceil<1+2/\eta$.
Since $\alpha\ge1$ and $d\ge0$, Definition~\ref{def:parameters} gives
\[
 M\ge2q(d+1)(1+2\alpha/\eta)\ge2q(1+2/\eta).
\]
The same definition gives $\tau\ge2M$, since $f(2M)\ge0$.
Hence
\[
 0<L=mw<3q(1+2/\eta)\le\frac32M\le\frac34\tau.
\]
Since $T=\tau+L$, this proves part~(c).
\end{proof}%

We will use part~(a) to place all repairs in $N_G^L[X]$ and part~(b)
to obtain ratio $\alpha+\eta$.
Part~(c) bounds the increase in the required local niceness scale.%

\subsection{\algname{ShiftedRepair}}\label{sec:overviewB}%

We give pseudocode for \algname{ShiftedRepair} ($\Bshift$) below.
\algname{SingleRepair} deletes the choices in $X$ from $A(G)$ and computes one
repair in each component of $N_G^{2\rho}[X]$.
\algname{ShiftedRepair} considers $m$ variants that delete choices from
neighbourhoods $N_G^{t_j}[X]$ of different radii.
For each variant, it computes repairs in the same enlarged zone
$N_G^L[X]$. In each component of this zone, it chooses a variant
whose union of retained choices and repair has minimum cardinality.%

\begin{algorithm}[H]
\caption{Algorithm $\Bshift$ (\algname{ShiftedRepair}).}
\label{alg:shifted-repair}
\begin{algorithmic}[1]
\color{black}
\Require An instance $G\in\D$ with a feasible solution that is
$T$-locally $\delta$-nice with respect to $\C$ and $\Pi$,
and a source algorithm $A$ as in Section~\ref{sec:guarantees},
with the parameters of Definition~\ref{def:shift-parameters}.
\Ensure A feasible set $\Bshift(G)$ with
$|\Bshift(G)|\le(\alpha+\eta)\OPT_\Pi(G)$.
\State Each vertex $u$ checks whether $G[N_G^\tau[u]]\in\C$.
Vertices for which the test fails form $X$.
\State Run $A$ on the whole graph $G$ for $r$ rounds and, for each
$j\in\{0,\ldots,m-1\}$, set
$S_j\gets A(G)\setminus N_G^{t_j}[X]$, where $t_j=jw+q$.
For each $j$, run the verifier on $S_j$ and compute
$R_j\gets\{v\in V(G):\neg F_G(v,S_j)\}$.
\State Set $Z\gets N_G^L[X]$.
For every $j$ and every connected component $H$ of $G[Z]$,
find a minimum-cardinality vertex set
$S'_{H,j}\subseteq N_G^\rho[R_j\cap H]$ feasible on $R_j\cap H$.
Break ties as in Definition~\ref{def:outputA}.
\State In each component $H$, choose an index $j$ minimizing
$|(S_j\cap H)\cup S'_{H,j}|$, taking the smallest $j$ at a tie.
For this index, set $C_H\gets(S_j\cap H)\cup S'_{H,j}$.
\State Return $\Bshift(G)\gets(A(G)\setminus Z)\cup\bigcup_H C_H$.
\end{algorithmic}
\end{algorithm}%

As in \algname{SingleRepair}, we start by detecting errors and running the source
algorithm. Each vertex $v$ tests whether $G[N_G^\tau[v]]$ belongs to $\C$.
The vertices with a negative answer form $X=X_\tau(G)$.
We run $A$ on the whole graph $G$, with all vertices participating,
and compute its output only once.

All variants use the same repair zone $Z\coloneqq N_G^L[X]$.
Let $\mathcal H$ be the family of vertex sets of the components of $G[Z]$.
In each $H\in\mathcal H$, we compare $m$ variants.
Variant $j$ keeps the set $S_j\coloneqq S(t_j)$ and requires repair on
the set $R_j\coloneqq R(t_j)$, where $S(t)$ and $R(t)$ are the families
of Lemma~\ref{lem:trusted} and Lemma~\ref{lem:zoneA}.
Thus
\[
 S_j=A(G)\setminus N_G^{t_j}[X],\qquad
 R_j=\{v\in V(G):\neg F_G(v,S_j)\}\qquad(0\le j<m).
\]

The set $S(0)=S$ is the retained set of \algname{SingleRepair}.
For example, for MDS, $R_j$ is the set of vertices not dominated by $S_j$.

For each index $j$ and each connected component $H$ of $G[Z]$,
we perform a search and obtain a repair $S'_{H,j}$.
We then compare the sizes of the unions of retained choices and repairs.
If a vertex belongs to both the repair and the set we keep, then it is
counted only once. The next definition specifies the choice rule.

\begin{definition}\label{def:outputB}
For each $H\in\mathcal H$ and $j\in\{0,\ldots,m-1\}$, let $S'_{H,j}$
be a minimum-cardinality subset of $N_G^\rho[R_j\cap H]$
that is feasible on $R_j\cap H$ in $G$.
Break ties by choosing the
lexicographically first increasing list of identifiers.
Let
\[
 C_{H,j}\coloneqq(S_j\cap H)\cup S'_{H,j}\qquad(0\le j<m).
\]
Let $C_H$ be a set of minimum cardinality among
$C_{H,0},\ldots,C_{H,m-1}$.
If several indices give this minimum, we take the smallest one.
We call this index the \textbf{index selected in $H$}.
The output is
\begin{equation*}
 \Bshift(G)\coloneqq(A(G)\setminus Z)\cup\bigcup_{H\in\mathcal H}C_H.
\end{equation*}
\end{definition}

Each component $H$ selects its own index and contributes
the set $C_H$ to the output.
Outside the repair zone $Z=N_G^L[X]$, which consists of all vertices at
distance at most $L$ from $X$, we keep the source choices.
If $X=\varnothing$, then $Z=\varnothing$, so there are no components
to repair and the output is $A(G)$.

To select the index in $H$, we compare the sizes of the sets $C_{H,j}$
computed for $H$. This uses only sets computed by the algorithm.
The comparison set $Y$ is used later, when we bound the cost of the output.
Once the required data for $H$ have been collected, both the search
for repairs and the comparison of all $m$ variants are local computations.
Section~\ref{sec:implementation} explains how to collect these data
and bounds the number of rounds.

\subsection{Correctness of \algname{ShiftedRepair}}\label{sec:witnessB}

We keep the notation $X,Z,\mathcal H,S_j,R_j$ from Section~\ref{sec:overviewB}.
Fix $\eta>0$ and let $T=T(\eta)=\tau+L$.
Let $G\in\D$ be an instance for which $\Pi$ has a feasible solution.
We assume that $G$ is $T$-locally $\delta$-nice with respect to $\C$ and $\Pi$,
where $\delta\ge0$ is the fixed integer diameter bound described in
Section~\ref{sec:guarantees}.
We first bound the diameters of the components of $G[Z]$ and establish
the existence of the repairs $S'_{H,j}$ from Definition~\ref{def:outputB}.
We then check feasibility of $\Bshift(G)$, allowing different components
to choose different indices.

\Needspace{5\baselineskip}
\subsubsection{The repair zone and existence of repairs}

\begin{lemma}\label{lem:absorb}
For nonnegative integers $a,b$,
\[
 N_G^b[X_a(G)]\subseteq X_{a+b}(G).
\]
Moreover, every component of $G[Z]$ has weak diameter at most $\delta$ in $G$.
\end{lemma}
\begin{proof}[\color{black}Proof]
For any $y\in N_G^b[X_a(G)]$, there exists $x\in X_a(G)$ such that
$\dist_G(x,y)\le b$. The triangle inequality gives
$N_G^a[x]\subseteq N_G^{a+b}[y]$.
Since $x\in X_a(G)$, the graph $G[N_G^a[x]]$ does not belong to $\C$.
By heredity, the graph $G[N_G^{a+b}[y]]$ containing it as an induced
subgraph also does not belong to $\C$. Thus $y\in X_{a+b}(G)$.

We apply this inclusion with $a=\tau$ and $b=L$ and obtain
\[
 Z=N_G^L[X_\tau(G)]\subseteq X_T(G)
 \subseteq N_G^{2\rho}[X_T(G)].
\]
Every component of $G[Z]$ therefore lies in one component of
$G[N_G^{2\rho}[X_T(G)]]$.
Since $G$ is $T$-locally $\delta$-nice,
every component of $G[N_G^{2\rho}[X_T(G)]]$ has weak diameter
at most $\delta$ in $G$, by Definition~\ref{def:nice}.
The same bound holds for the component of $G[Z]$ that it contains.
\end{proof}

The repair zone $Z=N_G^L[X]$ consists of all vertices at distance at most
$L$ from $X$. By Lemma~\ref{lem:absorb}, local niceness at scale
$T=\tau+L$ ensures that every component of $G[Z]$ has weak diameter
at most $\delta$ in $G$. The error set $X$ is still determined using
balls of radius $\tau$.
To locate the failed constraints and the vertices needed for repairs,
we use the general repair lemma already proved in
Lemma~\ref{lem:repair-exists}. We check its assumptions for every shift
in the following lemma.

\begin{lemma}\label{lem:shift-repairs}
For every $H\in\mathcal H$ and $j\in\{0,\ldots,m-1\}$,
the repair $S'_{H,j}$ of Definition~\ref{def:outputB} exists and lies in $H$.
\end{lemma}
\begin{proof}[\color{black}Proof]
Fix $j\in\{0,\ldots,m-1\}$.
We check Lemma~\ref{lem:repair-exists} with $t=t_j$,
$Z=N_G^L[X]$, and $T=\tau+L$.
By Fact~\ref{fact:shift-parameters}(a) and $j+1\le m$,
$t_j+2\rho=(j+1)w\le mw=L$. Hence
\[
 N_G^{t_j+2\rho}[X]=N_G^{(j+1)w}[X]\subseteq N_G^L[X]=Z.
\]
Then we apply Lemma~\ref{lem:absorb} with $a=\tau$ and $b=L$ and obtain
\[
 Z=N_G^L[X_\tau(G)]\subseteq X_{\tau+L}(G)
 =X_T(G)\subseteq N_G^{2\rho}[X_T(G)].
\]
Moreover, $\tau\ge r+\rho+1$ by Definition~\ref{def:parameters},
and $T=\tau+L\ge\tau$ because $L>0$.
Finally, $G$ is a feasible instance that is $T$-locally $\delta$-nice,
as assumed at the start of Section~\ref{sec:witnessB}.
Thus all assumptions of Lemma~\ref{lem:repair-exists} hold.
Its part~(c) shows that, for every $H\in\mathcal H$,
$N_G^\rho[R_j\cap H]\subseteq H$ and some subset of this neighbourhood
is feasible on $R_j\cap H$.
A minimum-cardinality such subset exists, so $S'_{H,j}$ exists
and is contained in $H$.
\end{proof}

If several minimum repairs exist for fixed $H,j$, we choose the one
whose increasing list of identifiers is lexicographically first.
If several indices minimize $|C_{H,j}|$, we choose the smallest index.
These rules uniquely determine $\Bshift(G)$.

\subsubsection{Feasibility of the combined output}

\begin{lemma}\label{lem:mixed}
The output $\Bshift(G)$ is feasible.
\end{lemma}
\begin{proof}[\color{black}Proof]
Let $t_{\max}\coloneqq t_{m-1}=L-2\rho$, as in
Fact~\ref{fact:shift-parameters}.
First observe that
\[
 A(G)\setminus N_G^{t_{\max}}[X]\subseteq\Bshift(G).
\]
Indeed, fix $u\in A(G)\setminus N_G^{t_{\max}}[X]$.
If $u\notin Z$, it remains in the output.
Otherwise, $u$ lies in some component $H\in\mathcal H$ of $G[Z]$.
Let $j$ be the index selected in $H$, so that $C_H=C_{H,j}$.
Since $t_j\le t_{\max}$, we have
$u\in S_j\cap H\subseteq C_{H,j}=C_H$.
Thus $u\in\Bshift(G)$ in both cases.

Now fix any $v\in V(G)$.
If $v\notin N_G^{t_{\max}+\rho}[X]$, then
$N_G^\rho[v]\cap N_G^{t_{\max}}[X]=\varnothing$.
The inclusion above gives
$A(G)\cap N_G^\rho[v]\subseteq\Bshift(G)$.
Since $v\notin X$, Lemma~\ref{lem:trusted} gives $F_G(v,A(G))$.
Locality and additivity therefore imply $F_G(v,\Bshift(G))$.

If instead $v\in N_G^{t_{\max}+\rho}[X]$, then
$N_G^\rho[v]\subseteq N_G^{t_{\max}+2\rho}[X]=Z$.
This ball is connected and contains $v$, so it lies in one component
$H\in\mathcal H$. Let $j$ be the index selected in $H$.
Recall that $R_j=\{u\in V(G):\neg F_G(u,S_j)\}$.
If $v\notin R_j$, then $F_G(v,S_j)$ holds by definition.
The inclusion
$S_j\cap N_G^\rho[v]\subseteq S_j\cap H
\subseteq C_{H,j}\subseteq\Bshift(G)$
and locality and additivity imply $F_G(v,\Bshift(G))$.
If $v\in R_j$, then $v\in R_j\cap H$, so $F_G(v,S'_{H,j})$ holds.
Since $S'_{H,j}\subseteq\Bshift(G)$, additivity again gives
$F_G(v,\Bshift(G))$.
The constraint at $v$ is satisfied in every case.
\end{proof}

\subsection{Total cost}

As in Lemma~\ref{lem:budgetA}, we first consider a region $U$ such that
every component of $G[Z]$ is either contained in $U$ or has no vertex in $U$.
We compare the output with a variant that uses one fixed index $j$
in all components.

\begin{fact}\label{fact:common-shift}
Let $U\subseteq V(G)$ satisfy $H\subseteq U$ or $H\cap U=\varnothing$
for every $H\in\mathcal H$.
For every $j\in\{0,\ldots,m-1\}$,
\[
 |\Bshift(G)\cap U|
 \le |S_j\cap U|
 +\sum_{\substack{H\in\mathcal H\\H\subseteq U}}|S'_{H,j}|.
\]
\end{fact}
\begin{proof}[\color{black}Proof]
Fix any $j\in\{0,\ldots,m-1\}$.
The output consists of $A(G)\setminus Z$ and the sets $C_H$
chosen in the components $H\in\mathcal H$.
The vertices of $A(G)\setminus Z$ that lie in $U$ form
$A(G)\cap(U\setminus Z)$.
For the sets $C_H$, recall that $C_H\subseteq H$.
If $H\subseteq U$, then $C_H\cap U=C_H$, so we count the whole set $C_H$.
If $H\cap U=\varnothing$, then $C_H\cap U=\varnothing$,
so this set contributes no vertices to the output in $U$.
By the assumption on $U$, these are the only two cases.
Thus the output in $U$ consists of $A(G)\cap(U\setminus Z)$
and the sets $C_H$ for components $H\subseteq U$.
The sets $C_H$ lie in different components inside $Z$,
so all these parts are pairwise disjoint. Hence
\[
 |\Bshift(G)\cap U|
 =|A(G)\cap(U\setminus Z)|
 +\sum_{\substack{H\in\mathcal H\\H\subseteq U}}|C_H|.
\]
By the choice of $C_H$ in Definition~\ref{def:outputB}, $|C_H|\le|C_{H,j}|$.
The size of the union defining $C_{H,j}$ is at most the sum of the
sizes of the two sets, so
\[
 |C_H|\le|S_j\cap H|+|S'_{H,j}|.
\]
Since $t_j\le L$, we have $N_G^{t_j}[X]\subseteq N_G^L[X]=Z$.
Thus $S_j=A(G)\setminus N_G^{t_j}[X]$ and $A(G)$ contain the same
vertices outside $Z$. Therefore,
$A(G)\cap(U\setminus Z)=S_j\cap(U\setminus Z)$.

Every vertex of $U\cap Z$ belongs to exactly one component
$H\in\mathcal H$. This component has a vertex in $U$,
so the assumption on $U$ implies $H\subseteq U$.
Summing $|S_j\cap H|$ over $H\subseteq U$ therefore counts all vertices
of $S_j\cap U$ inside $Z$, each exactly once.
The remaining vertices of $S_j\cap U$ lie outside $Z$ and number
$|A(G)\cap(U\setminus Z)|$.
Thus
\[
 |A(G)\cap(U\setminus Z)|
 +\sum_{\substack{H\in\mathcal H\\H\subseteq U}}|S_j\cap H|
 =|S_j\cap U|.
\]
Substituting the bounds on $|C_H|$ into the output count yields
\[
 \begin{aligned}
 |\Bshift(G)\cap U|
 &\le |A(G)\cap(U\setminus Z)|
 +\sum_{\substack{H\in\mathcal H\\H\subseteq U}}
   \bigl(|S_j\cap H|+|S'_{H,j}|\bigr)\\
 &=|S_j\cap U|
 +\sum_{\substack{H\in\mathcal H\\H\subseteq U}}|S'_{H,j}|.
 \end{aligned}\qedhere
\]
\end{proof}

The next lemma bounds both terms in Fact~\ref{fact:common-shift}.
We charge the repairs inside $N_G^{(j+1)w}[X]$
and the retained choices to vertices outside $N_G^{jw}[X]$.

\begin{lemma}\label{lem:shift-costs}
Let $U\subseteq V(G)$ satisfy $H\subseteq U$ or $H\cap U=\varnothing$
for every $H\in\mathcal H$, and let $Y\subseteq V(G)$ be feasible on
$N_G^k[U]$ in $G$.
For every $j\in\{0,\ldots,m-1\}$,
\begin{enumerate}[label=\textup{(\alph*)},leftmargin=*]
\item $\displaystyle
 \sum_{\substack{H\in\mathcal H\\H\subseteq U}}|S'_{H,j}|
 \le |Y\cap N_G^{(j+1)w}[X]|$,
\item $\displaystyle
 |S_j\cap U|\le\alpha_\theta\,|Y\setminus N_G^{jw}[X]|$.
\end{enumerate}
\end{lemma}
\begin{proof}[\color{black}Proof]
For part~(a), fix $H\in\mathcal H$ with $H\subseteq U$.
Since $R_j\cap H\subseteq U\subseteq N_G^k[U]$,
the set $Y$ is feasible on $R_j\cap H$.
By Observation~\ref{lem:restrict}, so is $Y\cap N_G^\rho[R_j\cap H]$.
The repair $S'_{H,j}$ has minimum cardinality among the sets allowed in
Definition~\ref{def:outputB}, so
\[
 |S'_{H,j}|\le|Y\cap N_G^\rho[R_j\cap H]|.
\]
For the fixed index $j$, we have $R(t_j)=R_j$.
Since $R_j\cap H\subseteq R_j$, Lemma~\ref{lem:repair-exists}(a),
applied with the parameters checked in Lemma~\ref{lem:shift-repairs}, gives
\[
 N_G^\rho[R_j\cap H]\subseteq N_G^\rho[R_j]
 \subseteq N_G^{t_j+2\rho}[X].
\]
Part~(c) of Lemma~\ref{lem:repair-exists} also gives
$N_G^\rho[R_j\cap H]\subseteq H$. Thus
\[
 N_G^\rho[R_j\cap H]\subseteq H\cap N_G^{t_j+2\rho}[X].
\]
These neighbourhoods are disjoint for different $H$ because they lie
in different components.
Each neighbourhood $N_G^\rho[R_j\cap H]$ is contained in
$N_G^{t_j+2\rho}[X]$.
By Fact~\ref{fact:shift-parameters}(a), $t_j+2\rho=(j+1)w$.
Thus all these neighbourhoods lie in $N_G^{(j+1)w}[X]$,
so their union also lies in this set.
Summing the repair bounds proves part~(a).

For part~(b), the assumption that $Y$ is feasible on $N_G^k[U]$
also makes it feasible on $N_G^k[S_j\cap U]$.
Since $S_j\cap U\subseteq A(G)\setminus X=S$,
Lemma~\ref{lem:charge} gives
\[
 |S_j\cap U|\le\alpha_\theta\,|Y\cap N_G^q[S_j\cap U]|.
\]
It remains to show that
$Y\cap N_G^q[S_j\cap U]\subseteq Y\setminus N_G^{jw}[X]$.
For this, we prove $N_G^q[S_j\cap U]\cap N_G^{jw}[X]=\varnothing$.
Suppose instead that a vertex $z$ lies in the intersection of
$N_G^q[S_j\cap U]$ and $N_G^{jw}[X]$.
Then some $u\in S_j\cap U$ satisfies $\dist_G(u,z)\le q$,
and the triangle inequality gives
\[
 \dist_G(u,X)\le\dist_G(u,z)+\dist_G(z,X)
 \le q+jw=t_j.
\]
This contradicts $u\in S_j=A(G)\setminus N_G^{t_j}[X]$,
and proves part~(b).
\end{proof}

The sets $Y\cap N_G^{(j+1)w}[X]$ and $Y\setminus N_G^{jw}[X]$
in Lemma~\ref{lem:shift-costs} overlap in
\[
 Y\cap\bigl(N_G^{(j+1)w}[X]\setminus N_G^{jw}[X]\bigr).
\]
For different indices $j$, these overlaps are pairwise disjoint.
Their total size is therefore at most $|Y|$.
At least one of the $m$ indices has overlap at most $|Y|/m$,
as recorded in the following fact.%

\begin{fact}\label{fact:thin-shell}
For every $Y\subseteq V(G)$, some $j\in\{0,\ldots,m-1\}$ satisfies
\[
 \bigl|Y\cap\bigl(N_G^{(j+1)w}[X]\setminus N_G^{jw}[X]\bigr)\bigr|
 \le |Y|/m.
\]
\end{fact}%

\begin{proof}
The neighbourhoods satisfy
\[
 N_G^0[X]\subseteq N_G^w[X]\subseteq\cdots\subseteq N_G^{mw}[X].
\]
Their consecutive differences
$N_G^{(i+1)w}[X]\setminus N_G^{iw}[X]$, for $i=0,\ldots,m-1$,
are therefore pairwise disjoint.
The sum of the sizes of their intersections with $Y$ is at most $|Y|$.
Since $m\ge1$, at least one of these sizes is at most $|Y|/m$.
\end{proof}

\begin{lemma}\label{lem:budget}
Let $U\subseteq V(G)$ satisfy $H\subseteq U$ or $H\cap U=\varnothing$
for every $H\in\mathcal H$.
Let $Y\subseteq V(G)$ be feasible on $N_G^k[U]$ in $G$. Then
\[
 |\Bshift(G)\cap U|\le(\alpha_\theta+1/m)|Y|.
\]
\end{lemma}
\begin{proof}[\color{black}Proof]
Choose $j$ by applying Fact~\ref{fact:thin-shell} to $Y$. Then
\[
 \bigl|Y\cap\bigl(N_G^{(j+1)w}[X]\setminus N_G^{jw}[X]\bigr)\bigr|
 \overset{\text{Fact~\ref{fact:thin-shell}}}{\le}|Y|/m.
\]%

Next, we apply Fact~\ref{fact:common-shift} with the chosen index $j$
and bound the two terms using Lemma~\ref{lem:shift-costs}. We obtain
\[
 \begin{aligned}
 |\Bshift(G)\cap U|
 &\overset{\text{Fact~\ref{fact:common-shift}}}{\le}
 |S_j\cap U|+\sum_{\substack{H\in\mathcal H\\H\subseteq U}}|S'_{H,j}|\\
 &\overset{\text{Lemma~\ref{lem:shift-costs}}}{\le}
 \alpha_\theta|Y\setminus N_G^{jw}[X]|+|Y\cap N_G^{(j+1)w}[X]|=(*).
 \end{aligned}
\]
Since $N_G^{jw}[X]\subseteq N_G^{(j+1)w}[X]$,
the set $Y\cap N_G^{(j+1)w}[X]$ is the disjoint union of
$Y\cap N_G^{jw}[X]$ and
$Y\cap(N_G^{(j+1)w}[X]\setminus N_G^{jw}[X])$.
We first split the second term of $(*)$ according to this union.
Since $\alpha_\theta\ge1$ by Fact~\ref{fact:alpha-theta},
we then increase the coefficient of $|Y\cap N_G^{jw}[X]|$
from $1$ to $\alpha_\theta$.
The first two terms now sum to $\alpha_\theta|Y|$.
Finally, the choice of $j$ gives
\[
 \begin{aligned}
 (*)
 &=\alpha_\theta|Y\setminus N_G^{jw}[X]|
   +|Y\cap N_G^{jw}[X]|
   +\bigl|Y\cap\bigl(N_G^{(j+1)w}[X]\setminus N_G^{jw}[X]\bigr)\bigr|\\
 &\overset{\mathclap{\text{Fact~\ref{fact:alpha-theta}}}}{\le}
   \alpha_\theta|Y\setminus N_G^{jw}[X]|
   +\alpha_\theta|Y\cap N_G^{jw}[X]|
   +\bigl|Y\cap\bigl(N_G^{(j+1)w}[X]\setminus N_G^{jw}[X]\bigr)\bigr|\\
 &=\alpha_\theta|Y|
   +\bigl|Y\cap\bigl(N_G^{(j+1)w}[X]\setminus N_G^{jw}[X]\bigr)\bigr|\\
 &\overset{\mathclap{\text{Fact~\ref{fact:thin-shell}}}}{\le}
   (\alpha_\theta+1/m)|Y|.
 \end{aligned}
\]
\end{proof}

The index $j$ in this proof is used only in the analysis.
The algorithm selects the index separately in each component.

\subsection{The guarantee for an arbitrary region}\label{sec:uniformB}

Lemma~\ref{lem:mixed} shows that $\Bshift(G)$ is a feasible solution.
Applying Lemma~\ref{lem:budget} to $U=V(G)$ and a global optimum,
and then using Fact~\ref{fact:shift-parameters}(b), gives the global
ratio $\alpha+\eta$.
We now establish the $(\delta+k)$-uniform $(\alpha+\eta)$-approximation guarantee.%

\begin{lemma}\label{lem:uniform-output}
For every $W\subseteq V(G)$, the output $\Bshift(G)$ satisfies
\[
 |\Bshift(G)\cap W|
 \le(\alpha+\eta)\OPT_\Pi(G,N_G^{\delta+k}[W]).
\]
\end{lemma}%

\begin{proof}
Apply Lemma~\ref{lem:uniform-extension} with
$B=\Bshift(G)$, $Z=N_G^L[X]$, and $c=\alpha_\theta+1/m$.
Lemma~\ref{lem:absorb} bounds the weak diameters of the components of
$G[Z]$ by $\delta$.
Fix any $U\subseteq V(G)$ such that each component of $G[Z]$ is either
contained in $U$ or has no vertex in $U$,
and choose a set $Y\subseteq V(G)$ of minimum size that is feasible on $N_G^k[U]$.
The condition on $U$ and the feasibility of $Y$ allow us to apply
Lemma~\ref{lem:budget}.
By Definition~\ref{def:feasible}, $|Y|=\OPT_\Pi(G,N_G^k[U])$,
which gives the required bound on $U$.
Thus
\[
 \begin{aligned}
 |\Bshift(G)\cap W|
 &\le(\alpha_\theta+1/m)\OPT_\Pi(G,N_G^{\delta+k}[W])\\
 &\le(\alpha+\eta)\OPT_\Pi(G,N_G^{\delta+k}[W]),
 \end{aligned}
\]
where the last inequality follows from Fact~\ref{fact:shift-parameters}(b).
\end{proof}

The next section bounds the number of rounds needed to compute both outputs.

\section{Main results and proofs}\label{sec:proofs}

We keep the assumptions and fixed source and target data of
Section~\ref{sec:guarantees}. We first state the round bounds for the
two algorithms. We then combine them with the feasibility and cost
bounds proved in Sections~\ref{sec:thmA} and~\ref{sec:thmB} to obtain
the two transfer theorems.

\subsection{Proofs of the transfer theorems}\label{sec:implementation}

Both algorithms collect the data needed in each component of the repair
zone. They then compute the repairs locally. For \algname{ShiftedRepair},
the same data also allow each component to compare all $m$ variants.

\begin{lemma}\label{lem:rounds}
In our model, the output $\Bsingle(G)$ from Definition~\ref{def:outputA}
can be computed in at most $\delta+\tau+r+3\rho+6$ rounds,
and the output $\Bshift(G)$ from Definition~\ref{def:outputB}
in at most $\delta+L+\tau+r+\rho+6$ rounds.
\end{lemma}
The proof appears in Appendix~\ref{app:rounds}.

\begin{maintheorem}\label{thm:A}
For every real $\eta>0$ and integer $\delta\ge0$, there is a
deterministic computable algorithm $\Bsingle$ that runs in at most
$\delta+\tau+r+3\rho+6$ rounds.

\noindent On every feasible instance $G\in\D$ that is
$\tau$-locally $\delta$-nice relative to $\C$ and $\Pi$
at the single scale $\tau=\tau(\eta)$, its output is feasible and,
for every $W\subseteq V(G)$,
\[
 |\Bsingle(G)\cap W|
 \le(\alpha+1+\eta)\OPT_\Pi(G,N_G^{\delta+k}[W]).
\]
\end{maintheorem}%

\begin{proof}[\color{black}Proof]
For $\Bsingle$ from Definition~\ref{def:outputA},
Lemma~\ref{lem:rounds} provides a local implementation within the stated
round bound. Lemma~\ref{lem:feasibleA} proves feasibility of the output,
and Lemma~\ref{lem:uniformA} proves the required uniform estimate.
\end{proof}

The bound holds for every region $W$, so the output is a
$(\delta+k)$-uniform $(\alpha+1+\eta)$-approximation.
Taking $W=V(G)$ gives the global ratio $\alpha+1+\eta$.
For fixed parameters, the round count is independent of the graph size.

\begin{maintheorem}\label{thm:B}
For every real $\eta>0$ and integer $\delta\ge0$, there is a
deterministic computable algorithm $\Bshift$ running in at most
$\delta+T+r+\rho+6$ rounds.

\noindent On every feasible instance $G\in\D$ that is $T$-locally
$\delta$-nice with respect to $\C$ and $\Pi$ at the single scale
$T=T(\eta)$, its output is feasible and satisfies, for every $W\subseteq V(G)$,
\[
 |\Bshift(G)\cap W|
 \le(\alpha+\eta)\OPT_\Pi(G,N_G^{\delta+k}[W]).
\]
\end{maintheorem}%

\begin{proof}[\color{black}Proof]
The algorithm $\Bshift$ in Algorithm~\ref{alg:shifted-repair} computes
the output of Definition~\ref{def:outputB}.
Lemma~\ref{lem:rounds} establishes the stated round bound.
Lemma~\ref{lem:mixed} proves feasibility, and
Lemma~\ref{lem:uniform-output} proves the required uniform estimate.
\end{proof}

\subsection{Transfer without errors}\label{sec:iteration}

Every tested ball may belong to the source class even when the whole
graph does not. For example, every ball of radius $\tau$ may be planar
while the whole graph is nonplanar. In this case, $X=\varnothing$,
so no repairs are needed. The following corollary extends the guarantee
of the same source algorithm to such target graphs, with ratio
$\alpha+\eta/2$ and the same round bound.
We keep all assumptions and notation of Section~\ref{sec:guarantees}.

\begin{corollary}\label{cor:iterate-no-errors}
Let $\eta>0$ and choose $\tau$ by Definition~\ref{def:parameters}.
If $G[N_G^\tau[v]]\in\C$, with inherited input data, for every
feasible $G\in\D$ and $v\in V(G)$, then $A$ is a $k$-uniform
$(\alpha+\eta/2)$-approximation on feasible instances in $\D$
in the same $r$ rounds.
\end{corollary}

\begin{proof}
We have $X=\varnothing$, so Lemma~\ref{lem:trusted} makes $S=A(G)$
feasible. For any $W\subseteq V(G)$, take a smallest set $Y$ feasible
on $N_G^k[W]$ and apply Lemma~\ref{lem:charge} to $J=S\cap W$.
Fact~\ref{fact:alpha-theta} gives
\[
 |A(G)\cap W|\le\alpha_\theta|Y|
 \le(\alpha+\eta/2)\OPT_\Pi(G,N_G^k[W]).\qedhere
\]
\end{proof}

The source algorithm may also come from an earlier transfer
if it satisfies the assumptions of Section~\ref{sec:guarantees}.

\section{Applications}\label{sec:applications}

\subsection{From planar graphs to bounded Euler genus}

For every fixed $g\ge0$ and $\eps>0$, our transfers give ratios
$\alpha+1+\eps$ and $\alpha+\eps$ on feasible instances of Euler
genus at most $g$. We use  a planar  $k$-uniform $\alpha$-approximation satisfying
Section~\ref{sec:guarantees} and assume that every planar induced ball is an allowed source input.
Both algorithms are computable and deterministic, take a constant
number of rounds, and need no embedding.

The two-dimensional control function $f_g$ for graphs of Euler genus
at most $g$ \cite{BBE24} lets us set $d=2$ and $f=f_g$ in our transfer
theorems. To apply them, we also need local niceness relative to the
planar source class. For $g\ge1$, \cite[Proposition~C.11]{BDGPSW26full}
shows that these graphs are $T$-locally $\delta(T)$-nice with
\[
 \delta(T)\coloneqq g\cdot(2T+4\rho+1).
\]
For $g=0$, every tested ball is planar, so the error set is empty
and the same formula gives $\delta(T)=0$.

We can now apply Theorems~\ref{thm:A} and~\ref{thm:B} with
$\eta=\eps$ to obtain the two ratios stated above.

Both algorithms  detect errors using balls of radius $\tau(\eps)$.
Their local niceness assumptions use different  scales,
$\tau(\eps)$ and $T(\eps)$, so  their uniform guarantees use 
$\delta(\tau(\eps))$ and $\delta(T(\eps))$  , respectively.

\par\addvspace{3.25ex plus 1ex minus .2ex}
\subsection{Cuttability and uniform source algorithms}\label{sec:cuttability}

Cuttability allows us to use an ordinary approximation as a uniform
source algorithm. We recall its definition from \cite{BDGPSW26}.

\begin{definition}\label{def:cuttable}
For $\beta\ge1$, a problem is \textbf{cuttable with parameter $\beta$
on $\C$} if, for every $G\in\C$ and $U\subseteq V(G)$,
there is $U\subseteq U'\subseteq V(G)$ such that
\[
 \OPT_\Pi(G[U'])\le\beta\,\OPT_\Pi(G,U).
\]
The induced instance inherits all input data.
\end{definition}

For a cuttable problem with parameter $\beta$ on a hereditary class
$\C$, an ordinary $r$-round $a$-approximation in our source model
is already an $(r+1)$-uniform $a\beta$-approximation
\cite[proof of Proposition~C.7]{BDGPSW26full}.
We can therefore use $\alpha=a\beta$ and $k=r+1$ in our transfers.
Under their respective target assumptions, Theorems~\ref{thm:A}
and~\ref{thm:B} give ratios $a\beta+1+\eta$ and $a\beta+\eta$.
The following two problems have $\beta=1$, so these ratios become
$a+1+\eta$ and $a+\eta$.

In distance-$h$ domination, every vertex must be within distance $h$
of a selected vertex. For fixed $h\ge1$, this problem is additive,
$h$-local, and cuttable with parameter one
\cite[Proposition~C.6]{BDGPSW26full}.
Our repairs use distances in the original graph and radius $\rho=h$.
In $p$-tuple domination, every vertex outside the selected set must
have at least $p$ selected neighbours. Selected vertices have no
such requirement, as in \cite{CHSWW14,CHSWW17}.
For fixed $p\ge2$, this problem is additive, 1-local, and cuttable
with parameter one \cite[Proposition~C.5]{BDGPSW26full}.
Here we use $\rho=1$.

\subsection{MDS from several source classes}\label{sec:seven-source}

Theorem~\ref{thm:B} allows us to transfer each MDS algorithm of Amiri
and Siebertz \cite{AS26} in Table~\ref{tab:mds-sources} to a wider target
class $\D$. This class must have a $d$-dimensional control function,
and its graphs must be $T$-locally $\delta$-nice relative to the
corresponding source class $\C$ and MDS, for fixed $\delta$ and the
radius $T$ required by the theorem. The table gives the resulting
approximation ratios. The transferred algorithms are deterministic
and take a constant number of rounds. The target graph does not have
to belong to $\C$, and the error set may be nonempty.

\begin{table}[ht]
\color{black}
\centering
\caption{MDS approximation ratios before and after transfer, for every $\eps>0$.
The target class $\D$ satisfies the assumptions above.}
\label{tab:mds-sources}
\begin{tabular}{@{}>{\raggedright\arraybackslash}p{.39\linewidth}ccl@{}}
\toprule
Source class $\C$ & \shortstack{On source\\class $\C$} &
\shortstack{After transfer\\to $\D$} & Source in \cite{AS26}\\
\midrule
Euler genus at most $g$, fixed $g\ge0$ & $7+\eps$ & $7+\eps$ & Theorem 8.4\\
Triangle-free planar graphs & $5+\eps$ & $5+\eps$ & Corollary 8.2\\
Treewidth at most $b$, fixed $b\ge1$ & $2b+1+\eps$ & $2b+1+\eps$ & Corollary 8.8(c)\\
$K_t$-minor-free graphs, fixed $3\le t\le9$ & $2t-3+\eps$ & $2t-3+\eps$ & Corollary 8.8(a)\\
\bottomrule
\end{tabular}
\end{table}

MDS is cuttable with parameter one on every hereditary class
\cite[proof of Claim~C.3]{BDGPSW26full}.
By Section~\ref{sec:cuttability}, each source algorithm is therefore
$(r+1)$-uniform with the same ratio.

To obtain a target ratio $a+\eps$ from a source family with ratio
$a+\zeta$, we choose the source loss $\zeta$ at most $\eps/2$ and
use transfer loss $\eta=\eps/2$. This explains why the same expression
appears in both columns of Table~\ref{tab:mds-sources}.
Each column states a guarantee available for every positive accuracy.

For a concrete target class, take any planar graph $P$ and attach
a clique of arbitrary size at one vertex $z$, identifying one clique
vertex with $z$. Let $\D_*$ contain all graphs obtained in this way.
This class contains arbitrarily large cliques, so its Euler genus is
unbounded. Nevertheless, the planar source applies through our
transfer, as the following argument shows.

Let $\pi:G\to P$ fix $P$ and send every added clique vertex to $z$.
For every two vertices $x,y$,
\[
 \dist_P(\pi(x),\pi(y))\le\dist_G(x,y)
 \le\dist_P(\pi(x),\pi(y))+2.
\]
A colouring of $P$ at scale $s$ therefore lifts along $\pi$ to a
colouring of $G$ whose monochromatic components in $G^s$ have weak
diameter at most $f_0(s)+2$, where $f_0$ is a two-dimensional control
function for planar graphs. Thus $\D_*$ has the common control
function $f_*(s)=f_0(s)+2$. For every integer $s\ge1$, a ball centred
outside $N_G^s[z]$ is planar. Hence the planar error set satisfies
$X_s(G)\subseteq N_G^s[z]$, and every component of
$G[N_G^2[X_s(G)]]$ has weak diameter at most $2s+4$.
The class is therefore $s$-locally $(2s+4)$-nice for MDS at every
such scale. Applying Theorem~\ref{thm:B} with the planar source gives
a deterministic constant-round $(7+\eps)$-approximation on $\D_*$.
When the attached clique has at least five vertices, the error set is
nonempty. This example verifies the target assumptions on a class
outside every fixed Euler genus bound.

\section{Comparison and limitations}\label{sec:limits}

Below, we compare the transfers from the three theorems in
Table~\ref{tab:comparison}. We then show what these guarantees give
for a concrete 2-domination example in Table~\ref{tab:scale-example}.
Meta-Theorem~I of \cite{BDGPSW26} requires the source class to be
hereditary and closed under disjoint unions, while our proofs require
only heredity. We denote the test radius of Meta-Theorem~I by
$T_0=f(2k+2\rho)+\max\{k+\rho,r\}$.

\begin{table}[ht]
\color{black}
\centering
\renewcommand{\arraystretch}{1.15}
\caption{Transfers of a $k$-uniform $\alpha$-approximation.}
\label{tab:comparison}
\begin{tabular}{@{}>{\raggedright\arraybackslash}p{.18\textwidth}
>{\raggedright\arraybackslash}p{.23\textwidth}
>{\raggedright\arraybackslash}p{.25\textwidth}
>{\raggedright\arraybackslash}p{.25\textwidth}@{}}
\toprule
Property & Meta-Theorem I \cite{BDGPSW26} & Theorem~\ref{thm:A} &
Theorem~\ref{thm:B}\\
\midrule
Ratio & $\alpha(d+1)+1$ & $\alpha+1+\eta$ & $\alpha+\eta$\\
Niceness scale & $T_0$ & $\tau(\eta)$ & $T(\eta)=\tau(\eta)+L(\eta)$\\
Uniformity radius & $\max\{k,\rho\}$ & $\delta+k$ & $\delta+k$\\
Rounds & $T_0+\delta+2$ & $\delta+\tau+r+3\rho+6$ &
$\delta+T+r+\rho+6$\\
\bottomrule
\end{tabular}
\end{table}

For the example in Table~\ref{tab:scale-example}, we consider
2-domination, where every vertex outside the selected set must have
at least two neighbours in the set. On planar graphs, Czygrinow and
coauthors give a 6-approximation in two rounds
\cite[Theorem~1]{CHSWW17}. Bonamy and coauthors transfer this source
to graphs of bounded Euler genus with ratio $19$
\cite[Corollary~4.5]{BDGPSW26full}. Using the same source,
Theorems~\ref{thm:A} and~\ref{thm:B} give ratios $7+\eta$ and
$6+\eta$, respectively, for every $\eta>0$ and every fixed genus bound.

\Needspace{3\baselineskip}
In the table, we fix the Euler genus bound at two and use
$\alpha=6$, $r=2$, $k=3$, and $\rho=1$.
Section~\ref{sec:cuttability} gives the required source uniformity.
As in Section~\ref{sec:applications}, we use $d=2$ and $\delta(s)=4s+10$.
We can choose $f_2(s)=cs$ for a fixed integer $c\ge1$
by \cite[Corollary~1.10]{BBE24}. Thus $T_0=8c+4$, while
$\tau(\eta)$ and $T(\eta)$ are chosen by
Definitions~\ref{def:parameters} and~\ref{def:shift-parameters}.
We evaluate $\delta$ at the scale required by each theorem.
The last two rows give numerical upper bounds on the scale divided by $T_0$,
valid for every such $c$.

\begin{table}[H]
\color{black}
\centering
\renewcommand{\arraystretch}{1.15}
\setlength{\tabcolsep}{5pt}
\caption{Transfer of the planar 6-approximation for 2-domination
to graphs of Euler genus at most two. Each round count is an upper bound.}
\label{tab:scale-example}
\begin{tabular}{@{}lccc@{}}
\toprule
Property & \shortstack{Meta-Theorem I\\\cite{BDGPSW26}} &
Theorem~\ref{thm:A} & Theorem~\ref{thm:B}\\
\midrule
Ratio, $\eta=1$ & $19$ & $8$ & $7$\\
Ratio, $\eta=0.1$ & $19$ & $7.1$ & $6.1$\\
Niceness scale & $T_0$ & $\tau(\eta)$ & $T(\eta)$\\
Uniformity radius & $3$ & $4\tau(\eta)+13$ & $4T(\eta)+13$\\
Rounds & $5T_0+12$ & $5\tau(\eta)+21$ & $5T(\eta)+19$\\
\midrule
Scale/$T_0$, $\eta=1$ & $1$ & $\le130$ & $\le131.2$\\
Scale/$T_0$, $\eta=0.1$ & $1$ & $\le1210$ & $\le1221.7$\\
\bottomrule
\end{tabular}
\end{table}

\Needspace{3\baselineskip}
Theorem~\ref{thm:A} improves the retained-cost analysis of
\cite[Algorithm~3]{BDGPSW26full} at a larger test radius.
Theorem~\ref{thm:B} also changes the repair rule, averaging the cost
shared by retained choices and repairs.
The price is an accuracy-dependent niceness scale and an output
uniformity radius that depends on $\delta$.

For $a\le b$, heredity gives $X_a(G)\subseteq X_b(G)$.
As the test radius grows, more vertices may become errors and components
of their $2\rho$-neighbourhood may merge.
A diameter bound at $T_0$ therefore does not by itself give a bound
at $\tau(\eta)$ or $T(\eta)$.
We require a common diameter bound at the scale used by each transfer.
Thus the smaller approximation loss comes with a stronger local niceness
assumption on the target class.

For a class that satisfies local niceness at every fixed scale,
such as graphs of bounded Euler genus relative to the planar source,
the larger scale does not change the class to which the theorem
applies. It increases the diameter bound $\delta$ and the resulting
round and uniformity bounds.

\Needspace{5\baselineskip}
\begin{fact}\label{fact:scale-comparison}
For the parameters of Definitions~\ref{def:parameters}
and~\ref{def:shift-parameters}, we have
\[
 \frac{\tau}{T_0}
 =\frac{f(2M)+2M}{f(2k+2\rho)+k+\rho},\qquad
 \frac{T}{T_0}
 =\frac{\tau}{T_0}
  +\frac{\lceil2/\eta\rceil(k+3\rho+1)}{f(2k+2\rho)+k+\rho}.
\]
For the same $\eta$, we also have $1<T/\tau<7/4$.
\end{fact}

\begin{proof}
The identities follow from Definitions~\ref{def:parameters}
and~\ref{def:shift-parameters}, using $k\ge r+1$ and $M\ge2q$.
The last inequality follows from Fact~\ref{fact:shift-parameters}(c).
\end{proof}

By Fact~\ref{fact:scale-comparison}, moving from \algname{SingleRepair}
to \algname{ShiftedRepair} increases the radius by less than $75\%$.
The source ratio $\alpha$ affects both radii through $M$.
Their increase over $T_0$ also depends on $f$, so no constant bounds
these scale factors for every allowed $f$.
The radii in Table~\ref{tab:scale-example} are sufficient for our guarantees.
Smaller radii may also work.

If $f(s)=cs$ for a fixed constant $c$, then
$\tau=2(c+1)M$ and $T<7\tau/4$. Both scales are consequently
$O_c\bigl(q(d+1)(1+\alpha/\eta)\bigr)$.
For the bounded-Euler-genus application, $\delta$ is linear in the
scale as well, so for a fixed source algorithm the round bounds grow
at most linearly in $1/\eta$.

Our transfers apply to additive local problems whose cost is the number
of selected vertices. Connected domination requires a global condition,
and independent domination is not additive. Weighted problems,
edge-set problems, maximization, and CONGEST need separate proofs.
The same applies to Meta-Theorems~II and~III of \cite{BDGPSW26}
for polynomially bounded identifiers.

Two questions remain. Can the ratio be improved at the original
scale $T_0$? Can a cuttability loss $\beta>1$ be removed?

\label{page:main-end}
\appendix
\section{Additional proofs}\label{app:proofs}
\subsection{Computing both outputs}\label{app:rounds}

\begin{proof}[Proof of Lemma~\ref{lem:rounds}]
Fix all parameters before execution. The control function is used
only through the integer $f(2M)$ in $\tau$; it need not be computable.
Put $\ell=2\rho$ for $\Bsingle$ and $\ell=L$ for $\Bshift$, and write
\[
 Z_\ell=N_G^\ell[X],\qquad
 r_\ell=\delta+\ell+\tau+r+\rho+6.
\]
In $r_\ell$ rounds, a vertex collects its extended $(r_\ell-1)$-view.

An error indicator needs the induced $\tau$-ball and its inherited
input. Membership in $Z_\ell$ can therefore be decided using
radius $\ell+\tau+1$.
A vertex outside $Z_\ell$ returns its source bit.
For $u\in Z_\ell$, its component $H$ in $G[Z_\ell]$ lies in
$N_G^\delta[u]$ by Lemma~\ref{lem:zoneA}(b) or Lemma~\ref{lem:absorb}.
Testing zone membership for all vertices of $N_G^{\delta+1}[u]$
uses radius $\delta+1+\ell+\tau+1$.
This identifies $H$ and verifies that no adjacent zone vertex is missing.
Distances here are in $G$; the communication paths may leave $H$.

For each deletion radius $t$ used by the algorithm, $t\le\ell$ and
\[
 x\in S(t)\quad\Longleftrightarrow\quad
 x\in A(G)\ \text{and}\ N_G^t[x]\cap X=\varnothing.
\]
Thus $u$ can evaluate the verifier at every $v\in H$ for every
retained-set variant, using the source bits in $N_G^\rho[v]$
and the relevant error indicators.
The required data lie within radius
$\delta+\rho+\ell+\tau+r+3<r_\ell-1$ of $u$.
They also determine each allowed repair set
$N_G^\rho[R(t)\cap H]\subseteq H$.

Enumerate its subsets and use the computable verifier to select
a minimum repair. Existence follows from Lemma~\ref{lem:zoneA}(c).
For $\Bshift$, also compare the sizes of
$(S_j\cap H)\cup S'_{H,j}$.
All vertices in $H$ have the same relevant data and identifiers,
so the fixed tie-breaking rules give the same repair and index.
These are finite local computations. Components compute in parallel,
since message size and local computation time are unrestricted.

On any other finite view, return a fixed default bit if the component
cannot be verified within the radius or no candidate repair exists.
The membership test, source rule, and verifier terminate on every
finite input, and the enumerations are finite. This defines a total
local rule with the required round bound.
On instances satisfying the assumptions, no default is needed.
The program uses neither $n$ nor an identifier bound, and computes
none of the colouring, partition, or comparison sets used in the analysis.
\end{proof}

\label{page:appendix-end}
\paragraph{Disclosure of AI use.}
The author used AI to improve the English and presentation
of the paper and to help check technical details of the proofs.


\begin{thebibliography}{BDG$^+$26b}\label{page:references-start}
\bibitem[AS26]{AS26}
S.~Amiri and S.~Siebertz.
\emph{Almost Optimal Constant-Round Approximation of Dominating Set
in Graph Classes with Excluded Minors}.
Preprint, 2026.
\href{https://arxiv.org/abs/2610.09618v2}{arXiv:2610.09618v2}.

\bibitem[BBE$^+$24]{BBE24}
M.~Bonamy, N.~Bousquet, L.~Esperet, C.~Groenland, C.-H.~Liu,
F.~Pirot, and A.~Scott.
\emph{Asymptotic dimension of minor-closed families and
Assouad--Nagata dimension of surfaces}.
Journal of the European Mathematical Society, 26, 3739--3791, 2024.
\href{https://doi.org/10.4171/JEMS/1341}{doi:10.4171/JEMS/1341}.

\bibitem[BDG$^+$26a]{BDGPSW26}
M.~Bonamy, A.~Das, C.~Gavoille, T.~Picavet, J.~Suomela, and A.~Wesolek.
\emph{Meta-Theorems for Cuttable Distributed Problems}.
PODC 2026, pp.~382--389.
\href{https://doi.org/10.1145/3796701.3815958}{doi:10.1145/3796701.3815958}.

\bibitem[BDG$^+$26b]{BDGPSW26full}
M.~Bonamy, A.~Das, C.~Gavoille, T.~Picavet, J.~Suomela, and A.~Wesolek.
\emph{Meta-Theorems for Cuttable Distributed Problems}.
Full version, 2026.
\href{https://arxiv.org/abs/2605.19157v1}{arXiv:2605.19157v1}.

\bibitem[BGPW25]{BGPW25}
M.~Bonamy, C.~Gavoille, T.~Picavet, and A.~Wesolek.
\emph{Distributed Approximation Algorithms for Minimum Dominating Set
in Locally Nice Graphs}.
Preprint, 2025.
\href{https://arxiv.org/abs/2507.04960}{arXiv:2507.04960}.

\bibitem[CHS$^+$14]{CHSWW14}
A.~Czygrinow, M.~Ha\'n\'ckowiak, E.~Szyma\'nska,
W.~Wawrzyniak, and M.~Witkowski.
\emph{Distributed Local Approximation of the Minimum
$k$-Tuple Dominating Set in Planar Graphs}.
OPODIS 2014, LNCS 8878, pp.~49--59.
\href{https://doi.org/10.1007/978-3-319-14472-6_4}{doi:10.1007/978-3-319-14472-6\_4}.

\bibitem[CHS$^+$17]{CHSWW17}
A.~Czygrinow, M.~Ha\'n\'ckowiak, E.~Szyma\'nska,
W.~Wawrzyniak, and M.~Witkowski.
\emph{Improved distributed local approximation algorithm
for minimum 2-dominating set in planar graphs}.
Theoretical Computer Science, 662, 1--8, 2017.
\href{https://doi.org/10.1016/j.tcs.2016.12.001}{doi:10.1016/j.tcs.2016.12.001}.

\bibitem[CHW08]{CHW08}
A.~Czygrinow, M.~Ha\'n\'ckowiak, and W.~Wawrzyniak.
\emph{Fast Distributed Approximations in Planar Graphs}.
DISC 2008, LNCS 5218, pp.~78--92.
\href{https://doi.org/10.1007/978-3-540-87779-0_6}{doi:10.1007/978-3-540-87779-0\_6}.

\bibitem[CHWW18]{CHWW18}
A.~Czygrinow, M.~Ha\'n\'ckowiak, W.~Wawrzyniak, and M.~Witkowski.
\emph{Distributed Approximation Algorithms for the Minimum
Dominating Set in $K_h$-Minor-Free Graphs}.
ISAAC 2018, LIPIcs 123, pp.~22:1--22:12.
\href{https://doi.org/10.4230/LIPIcs.ISAAC.2018.22}{doi:10.4230/LIPIcs.ISAAC.2018.22}.

\bibitem[GKM17]{GKM17}
M.~Ghaffari, F.~Kuhn, and Y.~Maus.
\emph{On the Complexity of Local Distributed Graph Problems}.
STOC 2017, pp.~784--797.
\href{https://doi.org/10.1145/3055399.3055471}{doi:10.1145/3055399.3055471}.

\bibitem[HKO$^+$25]{HKOdMSV25}
O.~Heydt, S.~Kublenz, P.~Ossona de Mendez, S.~Siebertz, and A.~Vigny.
\emph{Distributed domination on sparse graph classes}.
European Journal of Combinatorics, 123, Article 103773, 2025.
\href{https://doi.org/10.1016/j.ejc.2023.103773}{doi:10.1016/j.ejc.2023.103773}.
Full text at \href{https://arxiv.org/abs/2207.02669}{arXiv:2207.02669}.

\bibitem[HLS14]{HLS14}
M.~Hilke, C.~Lenzen, and J.~Suomela.
\emph{Brief announcement: Local approximability of minimum dominating
set on planar graphs}.
PODC 2014, pp.~344--346.
\href{https://doi.org/10.1145/2611462.2611504}{doi:10.1145/2611462.2611504}.


















\bibitem[HSV22]{HSV22}
O.~Heydt, S.~Siebertz, and A.~Vigny.
\emph{Local planar domination revisited}.

SIROCCO 2022, LNCS 13298, pp.~154--173.
\href{https://doi.org/10.1007/978-3-031-09993-9_9}{doi:10.1007/978-3-031-09993-9\_9}.





























Full version at \href{https://arxiv.org/abs/2111.14506}{arXiv:2111.14506}.



\bibitem[KMW16]{KMW16}
F.~Kuhn, T.~Moscibroda, and R.~Wattenhofer.
\emph{Local Computation: Lower and Upper Bounds}.


Journal of the ACM, 63(2), Article 17, 2016.





























\href{https://doi.org/10.1145/2742012}{doi:10.1145/2742012}.

\bibitem[LPW13]{LPW13}
C.~Lenzen, Y.-A.~Pignolet, and R.~Wattenhofer.
\emph{Distributed minimum dominating set approximations in
restricted families of graphs}.
Distributed Computing, 26(2), 119--137, 2013.
\href{https://doi.org/10.1007/s00446-013-0186-z}{doi:10.1007/s00446-013-0186-z}.

\bibitem[RG20]{RG20}
V.~Rozho\v{n} and M.~Ghaffari.
\emph{Polylogarithmic-Time Deterministic Network Decomposition
and Distributed Derandomization}.
STOC 2020, pp.~350--363.
\href{https://doi.org/10.1145/3357713.3384298}{doi:10.1145/3357713.3384298}.
Full text at \href{https://arxiv.org/abs/1907.10937v2}{arXiv:1907.10937v2}.

\bibitem[Waw14]{Waw14}
W.~Wawrzyniak.
\emph{A strengthened analysis of a local algorithm for the
minimum dominating set problem in planar graphs}.
Information Processing Letters, 114(3), 94--98, 2014.
\href{https://doi.org/10.1016/j.ipl.2013.11.008}{doi:10.1016/j.ipl.2013.11.008}.

\bibitem[Waw26]{Wawrzyniak26}
W.~Wawrzyniak.
\emph{Near-Optimal Distributed Domination in Planar Graphs}.
Preprint, 2026.
\href{https://arxiv.org/abs/2609.32736v1}{arXiv:2609.32736v1}.

\label{page:references-end}
\end{thebibliography}
\end{document}